\newif\ifarxiv
\arxivtrue

\documentclass[a4paper,UKenglish,cleveref, autoref, thm-restate]{lipics-v2021}

\pdfoutput=1 
\hideLIPIcs  

\nolinenumbers

\usepackage{amsmath}
\usepackage{cite}

\usepackage{cleveref}
\usepackage{amsthm}
\usepackage{amssymb}
\usepackage{amsfonts}
\usepackage{mathrsfs}
\usepackage{wrapfig} 

\usepackage{thmtools} 
\usepackage{thm-restate}

\usepackage{caption} 
\usepackage{subcaption}
\usepackage{changepage} 
\usepackage{mathtools} 

\usepackage{tikz}

\usetikzlibrary{arrows,decorations.pathmorphing,decorations.pathreplacing,backgrounds,positioning,fit,matrix}
\usetikzlibrary{shapes,calc,patterns,arrows.meta}
\usetikzlibrary{external}
\tikzset{
	vert/.style={circle,inner sep=1.5,fill=white,draw=black,minimum size=.3cm},
    dummy/.style={circle,fill=black,draw=black,inner sep=2.5},
	edge/.style={color=black, thick},
	diredge/.style={->,>={Stealth[width=8pt,length=8pt]},color=black, thick},
	timelabel/.style={fill=white,font=\footnotesize, text centered},
	wave/.style={decorate,decoration={coil,aspect=0}},
	dirwave/.style={->, >={Stealth[width=8pt,length=8pt]},decorate,decoration={coil,aspect=0}},
	diredge2/.style={->,>={Stealth[width=8pt,length=8pt]}}
}

\usepackage{enumerate}
\usepackage{todonotes}
\usepackage{comment}

\usepackage{algorithm}
\usepackage[noend]{algpseudocode}

\usepackage{xspace}

\usepackage[T1]{fontenc}
\usepackage[utf8]{inputenc}
\usepackage{lmodern}

\usepackage{hyperref}
\usepackage{multicol}

\crefname{claim}{Claim}{Claims}
\crefname{observation}{Observation}{Observations}

\newcommand{\ie}{i.\,e.,\ }

\newcommand{\degtwo}{degree-2\xspace}

\newcommand{\vc}{\mathrm{vc}}
\newcommand{\COL}[1][3]{\textsc{#1-Coloring}\xspace}

\usepackage{tabularx}
\newcommand{\problemdef}[3]{
	\begin{center}
		\begin{minipage}{0.95\textwidth}
			\noindent
			#1
			\vspace{5pt}\\
			\setlength{\tabcolsep}{3pt}
			\begin{tabularx}{\textwidth}{@{}lX@{}}
				\textbf{Input:}& #2 \\
				\textbf{Question:}& #3
			\end{tabularx}
		\end{minipage}
	\end{center}
}

\newcommand{\optproblemdef}[3]{
	\begin{center}
		\begin{minipage}{0.95\textwidth}
			\noindent
			#1
			\vspace{5pt}\\
			\setlength{\tabcolsep}{3pt}
			\begin{tabularx}{\textwidth}{@{}lX@{}}
				\textbf{Input:}& #2 \\
				\textbf{Task:}& #3
			\end{tabularx}
		\end{minipage}
	\end{center}
}

\newcommand{\DeltaUpperBoundLong}{\textsc{Periodic Upper-Bounded Temporal Tree Realization}\xspace}
\newcommand{\deltaUpperBound}{\textsc{TTR}\xspace}

\title{Realizing temporal transportation trees} 

\author{George B. Mertzios}{Department of Computer Science, Durham University, UK}{george.mertzios@durham.ac.uk}{https://orcid.org/0000-0001-7182-585X}{Supported by the EPSRC grant EP/P020372/1.}

\author{Hendrik~Molter}{Department of Computer Science, Ben-Gurion~University~of~the~Negev, 
	Beer-Sheva, 
	Israel}{molterh@post.bgu.ac.il}{https://orcid.org/0000-0002-4590-798X}{Supported by the Israel Science Foundation, grant nr.~1470/24, by the European Union's Horizon Europe research and innovation programme under grant agreement 949707, and by the European Research Council, grant nr.~101039913 (PARAPATH).}
	
\author{Nils Morawietz}{Friedrich Schiller University Jena, Institute of Computer Science, Germany}{nils.morawietz@uni-jena.de}{https://orcid.org/0000-0002-7283-4982}{Partially supported by the French ANR, project ANR-22-CE48-0001 (TEMPOGRAL).}

\author{Paul G. Spirakis}{Department of Computer Science, University of Liverpool, UK}{p.spirakis@liverpool.ac.uk}{https://orcid.org/0000-0001-5396-3749}{Supported by the EPSRC grant EP/P02002X/1.}

\authorrunning{George B. Mertzios, Hendrik Molter, Nils Morawietz, and Paul G. Spirakis} 

\Copyright{George B. Mertzios, Hendrik Molter, Nils Morawietz, and Paul G. Spirakis} 



\relatedversion{} 

\acknowledgements{The authors want to thank Nina Klobas for discussions on the problem that lead to some initial results.}

\begin{document}
\maketitle

\begin{abstract}
In this paper, we study the complexity of the \textit{periodic temporal graph realization} problem with respect to 
\textit{upper bounds} on the fastest path durations among its vertices. 
This constraint with respect to upper bounds appears naturally in \textit{transportation network design} applications where, for example, a road network is given, and the goal is to appropriately schedule periodic travel routes, while not exceeding some desired upper bounds on the travel times. 
In our work, we focus only on underlying \textit{tree topologies}, which are fundamental in many transportation network applications.

As it turns out, the periodic upper-bounded temporal tree realization problem (\deltaUpperBound) has a very different computational complexity behavior than both (i)~the classic graph realization problem with respect to shortest path distances in static graphs and (ii)~the periodic temporal graph realization problem with \textit{exact} given fastest travel times (which was recently introduced). 
First, we prove that, surprisingly, \deltaUpperBound is NP-hard, even for a constant period $\Delta$ and when the input tree $G$ satisfies at least one of the following conditions: 
(a)~$G$ is a star, or (b)~$G$ has constant maximum degree. 
Second, we prove that \deltaUpperBound is fixed-parameter tractable (FPT) with respect to the number of leaves in the input tree~$G$, via a novel combination of techniques for totally unimodular matrices and mixed integer linear programming.

\end{abstract}

\section{Introduction}\label{intro-sec}
A temporal (or dynamic) graph is a graph whose topology is subject to discrete changes over time. 
This paradigm reflects the structure and operation of a great variety of modern networks; 
social networks, wired or wireless networks whose links change dynamically, transportation networks, and several physical systems are only a few examples of networks that change over time~\cite{michailCACM,Nicosia-book-chapter-13,holme2019temporal}. 
Inspired by the foundational work of Kempe et al.~\cite{KKK00}, we adopt here a simple model for
temporal graphs, in which the vertex set remains unchanged while each edge is equipped with a set of integer time labels.
	
\begin{definition}[temporal graph~\cite{KKK00}]
\label{temp-graph-def} A \emph{temporal graph} is a pair $(G,\lambda)$,
where $G=(V,E)$ is an underlying (static) graph and $\lambda :E\rightarrow 2^\mathbb{N}$ is a \emph{labeling} function which assigns to every edge of $G$ a finite set of discrete time labels.
\end{definition}

Here, whenever $t \in \lambda(e)$, we say that the edge $e$ is \emph{active} or \emph{available} at time $t$. In the context of temporal graphs, where the notion of vertex adjacency is time-dependent, the notions of path and distance also need to be redefined. The most natural temporal analogue of a path is that of a \emph{temporal} (or \emph{time-dependent}) path, which is 
a path of the underlying static graph whose time labels strictly increase along the edges of the path. 
The \textit{duration} of a temporal path is the number of discrete time steps needed to traverse it. Finally, a temporal path from vertex $u$ to vertex $v$ is \textit{fastest} if it has the smallest duration among all temporal paths from $u$ to $v$.

The \emph{graph realization} problem with respect to some graph property $\mathcal{P}$ is to compute a graph that satisfies $\mathcal{P}$, or to decide that no such graph exists. 
The main motivation for graph realization problems stems both from network design and from ``verification'' applications in engineering. 
In \emph{network design} applications, the goal is to design network topologies having a desired property~\cite{augustine2022distributed,grotschel1995design}. 
On the other hand, in \emph{verification} applications, given the outcomes of some exact experimental measurements, 
the aim is to (re)construct an input network that complies with them. 
If such a reconstruction is not possible, this proves that the measurements are incorrect or implausible (resp.~that the algorithm making the computations is incorrectly implemented).

One example of a graph realization problem is the recognition of probe interval graphs, in the context of the physical mapping of DNA, see~\cite{McMorris98,McConnellS02} and~\cite[Chapter 4]{GolumbicTrenk04}.
Analyzing the computational complexity of the graph realization problems for various natural and fundamental graph properties $\mathcal{P}$ requires a deep understanding of these properties.
Among the most studied parameters for graph realization 
are constraints on the distances between vertices~\cite{barNoy2022GraphRealization,barNoy2021composed,hakimi1965distance,Patrinos-Hakimi-72,Rubei16,Tamura93,chung2001distance,bixby1988almost,culberson1989fast}, 
on the vertex degrees~\cite{GolovachM17,gomory1961multi,hakimi1962realizability,Bar-NoyCPR20,erdos1960graphs}, 
on the eccentricities~\cite{barNoy2020efficiently,hell2009linear,behzad1976eccentric,lesniak1975eccentric}, and on connectivity~\cite{fulkerson1960zero,frank1992augmenting,chen1966realization,frank1994connectivity,frank1970connectivity,gomory1961multi}, among others. 
Although the majority of studies of graph realization problems concern \textit{static} graphs, the \textit{temporal} (periodic) graph realization problem with respect to given (exact) delays of the fastest temporal paths among vertices has been recently studied in~\cite{EMW24,KMMS23TempRealiz}, motivated by \textit{verification} applications where exact measurements are given.

In this paper, we consider the (periodic) temporal graph realization problem with respect to the durations of the fastest temporal paths from a \textit{network design} perspective, \ie where only \textit{upper bounds} on the durations of fastest temporal paths are given. 
One important application domain of network design problems is that of \textit{transportation network} design where, for example, a road network is given, and the goal is to appropriately schedule periodic travel routes while not exceeding some desired upper bounds on travel times. 
Our work is motivated by the fact that in many transportation network applications the  underlying graph has a \textit{tree} structure~\cite{PhysRevLett00,Barthelemy_06,PhysRevLett98}. For example, many airlines or railway transportation networks are even star graphs (having a big hub at the center of the network, e.g.~in the capital city). 
The formal definition of our problem is as follows (see \Cref{sec:prelims} for formal definitions of all used terminology).

\problemdef{\DeltaUpperBoundLong (\deltaUpperBound)}
{A tree $G=(V,E)$ with $V=\{v_1,v_2,\ldots,v_n\}$, an $n \times n$ matrix~$D$ of positive integers, and a positive integer $\Delta$.}
{Does there exist a $\Delta$-periodic labeling $\lambda: E \rightarrow \{1,2,\ldots,\Delta\}$ such that, 
	for every~$i,j$, the duration of the fastest temporal path from $v_i$ to $v_j$ in the $\Delta$-periodic temporal graph $(G,\lambda,\Delta)$ is \emph{at most} $D_{i,j}$.}

Many natural and technological systems exhibit a periodic temporal behavior. 
This is true even more on transportation networks~\cite{Arrighi2023Multi}, where the goal is to build periodic schedules of transportation units (e.g.~trains, buses, or airplanes). 
The most natural constraint, while designing such periodic transportation schedules, is that the fastest travel time (\ie the duration of the fastest temporal path) between a specific pair of vertices does not exceed a specific desired upper bound. That is, if the travel time between $u$ and $v$ in the resulting schedule is shorter than this upper bound, the schedule is even better (and thus, still feasible). 
Periodic temporal graphs have also been studied in various different contexts (see~\cite[Class~8]{casteigts2012time} and~\cite{Arrighi2023Multi,ErlebachS20,ErlebachMSW24,morawietz2021timecop,morawietz2020timecop}).

We focus on the most fundamental case of periodic temporal graphs, where every edge has the same period $\Delta$ and it appears exactly once within each period. 
As it turns out, \deltaUpperBound has a very different computational complexity behavior than both (i)~the classic graph realization problem with respect to shortest path distances in static graphs~\cite{hakimi1965distance} 
and (ii)~the periodic temporal graph realization problem with \textit{exact} given fastest travel times~\cite{KMMS23TempRealiz}. We remark that very recently, the case where every edge can appear multiple times in each $\Delta$ period has been studied~\cite{EMW24}.

\subparagraph{Our contribution.}
Our results in this paper are given in two main directions. 
First, we prove that \deltaUpperBound is NP-hard, even for a constant period $\Delta$ and when the input tree $G$ satisfies at least one of the following conditions: 
(a)~$G$ is a star (for any $\Delta\geq 3$), or (b)~$G$ has diameter at most $6$ and constant pathwidth (even for $\Delta=2$), or (c)~$G$ has maximum degree at most $8$ and constant pathwidth  (even for $\Delta= 2$). 
Note that, for $\Delta=1$, \deltaUpperBound becomes trivial (as in this case fastest paths coincide with shortest paths in the underlying tree). 
On the one hand, our hardness results come in wide contrast to the classic (static) graph realization problem with respect to shortest path distances. Indeed, this static analogue is easily solvable in polynomial time in two steps~\cite{hakimi1965distance}. 
On the other hand, the complexity of \deltaUpperBound is also surprisingly different than the periodic temporal graph realization problem with \textit{exact} given fastest travel times. More specifically, when the input integer matrix $D$ gives the exact fastest travel times, the problem becomes solvable in polynomial time on trees~\cite{KMMS23TempRealiz}. 
Note that our hardness results rule out the existence of FPT-algorithms for \deltaUpperBound for all reasonable parameter on trees besides the number of leaves (assuming P $\neq$ NP).

Second, we prove that \deltaUpperBound is fixed-parameter tractable (FPT) with respect to the number of leaves in the input tree $G$. 
That is, long chains of vertices of degree 2 do not affect the complexity of the problem. 
To provide our FPT algorithm, we reduce \deltaUpperBound to a number \textsc{Mixed Integer Linear Program} (\textsc{MILP}) instances that is upper-bounded by a function of the number of leaves in the input tree. Furthermore, the number of integer variables in each \textsc{MILP} instance is also upper-bounded by a function of the number of leaves in the input tree.
This allows us to use a known FPT algorithm for MILP parameterized by the number of integer variables~\cite{Lenstra1983Integer,dadush2011enumerative} to solve all \textsc{MILP} instances in FPT-time with respect to the number of leaves of the input tree. We prove that the \deltaUpperBound instance is a yes-instance if and only if at least one of the MILP instances admits a feasible solution. 
To this end, we use a novel combination of techniques for totally unimodular matrices and mixed integer linear programming to design the \textsc{MILP} instances in a way that if they admit a feasible solution, then we can assume that all variables are set to integer values.

\subparagraph{Related work.} Since the 1960s, (static) graph realization problems have been extensively studied. We give an overview of related work in the introduction. 

Several problems related to graph realization have been studied in the temporal setting. In most cases, the goal is to assign labels to (sets of) edges of a given static graph to achieve certain connectivity-related properties~\cite{KlobasMMS22,MertziosMS19,akrida2017complexity,enright2021assigning,EMW24,KMMS23TempRealiz,EMW24}. In most of the mentioned works, the sought labeling is not periodic and the desired property is to establish temporal connectivity among all vertices~\cite{KlobasMMS22,MertziosMS19,akrida2017complexity} or a subset of vertices~\cite{KlobasMMS22}, or to restrict reachability~\cite{enright2021assigning}. 
Among all these works, periodic labelings and the duration of the temporal paths were only considered in~\cite{KMMS23TempRealiz,EMW24}.

Numerous further problems related to temporal connectivity have been studied~\cite{Mertzios-transitivity21,enright2021deleting,MolterRZ21,deligkas2022optimizing,erlebach2021temporal,Flu+19a,Zsc+19,EnrightMM23,casteigts2021finding,FMNR22a,meusel2025directedtemporaltreerealization,EMM25}. Due to the added temporal dimension, the considered problems are typically computationally harder than their non-temporal counterparts, and some of which do not even have a non-temporal counterpart. In many cases, restricting the underlying topology to trees has been considered to obtain tractability results~\cite{KMZ23,klobas2023interference,Akrida-explorer-21}.

\begin{figure}[t]
    \centering
    \ifarxiv
    \includegraphics[scale=1]{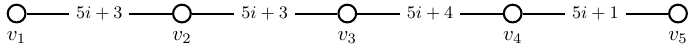}
    \else
	\begin{tikzpicture}[line width=1pt,xscale=1.4]
		\node[vert,label=below:$v_1$] (1) at (1,0) {};
		\node[vert,label=below:$v_2$] (2) at (3,0) {};
		\node[vert,label=below:$v_3$] (3) at (5,0) {};
		\node[vert,label=below:$v_4$] (4) at (7,0) {};
		\node[vert,label=below:$v_5$] (5) at (9,0) {};
		\draw (1) --node[timelabel]  {$5i+3$} (2) --node[timelabel] {$5i+3$}  (3) --node[timelabel] {$5i+4$} (4) --node[timelabel] {$5i+1$} (5);	
	\end{tikzpicture}
	\fi
	\caption{Visualization of a $\Delta$-periodic temporal graph $(G,\lambda,\Delta)$ with $\Delta = 5$ and the following $\Delta$-periodic labeling $\lambda: E \rightarrow \{1,2,\ldots,5\}$: $\lambda(\{v_1, v_2\})=\lambda(\{v_2, v_3\})=3$, $\lambda(\{v_3, v_4\})=4$, and~$\lambda(\{v_4, v_5\})=1$. 
    A fastest temporal path from $v_1$ to $v_5$ first traverses $\{v_1,v_2\}$ at time $3$, then $\{v_2,v_3\}$ a time $8$, then $\{v_3,v_4\}$ at time $9$, and then $\{v_4,v_5\}$ at time $11$, and has duration $9$.
	\label{fig:periodic-example}}
\end{figure}

\section{Preliminaries and notation}\label{sec:prelims}

An undirected graph~$G=(V,E)$ consists of a set~$V$ of vertices 
and a set~$E \subseteq \binom{V}{2}$ of edges.
We denote by~$V(G)$ and~$E(G)$ the vertex and edge set of~$G$, respectively.

Let $G=(V,E)$ and $\Delta\in \mathbb{N}$, and let $\lambda: E \rightarrow \{1,2,\ldots,\Delta\}$ be an edge-labeling function that assigns to every edge of $G$ exactly one of the labels from $\{1,\ldots,\Delta\}$. 
Then we denote by $(G,\lambda,\Delta)$ the \emph{$\Delta$-periodic temporal graph} $(G,L)$, where for every edge $e\in E$ we have $L(e)=\{\lambda(e) + i\Delta \mid i\geq 0\}$. 
In this case, we call $\lambda$ a \emph{$\Delta$-periodic labeling} of~$G$. 
When it is clear from the context, we drop $\Delta$ and denote the ($\Delta$-periodic) temporal graph by $(G,\lambda)$.

A  \emph{temporal $(s,z)$-walk} (or \emph{temporal walk}) of length~$k$ from vertex $s=v_0$ to vertex $z=v_k$ in a $\Delta$-periodic temporal graph~$(G,L)$ is a sequence $P = \left(\left(v_{i-1},v_i,t_i\right)\right)_{i=1}^k$ of triples that we call \emph{transitions},  such that for all $i\in[k]$ we have that $t_i\in L(\{v_{i-1},v_i\})$ and for all $i\in [k-1]$ we have that $t_i < t_{i+1}$.
Moreover, we call $P$ a \emph{temporal $(s,z)$-path} (or \emph{temporal path})  
of length~$k$ if~$v_i\neq v_j$ for all~$i, j\in \{0,\ldots,k\}$ with $i\neq j$.
Given a temporal path $P=\left(\left(v_{i-1},v_i,t_i\right) \right)_{i=1}^k$, we denote the set of vertices of $P$ by $V(P)=\{v_0,v_1,\ldots,v_k\}$.
A temporal $(s,z)$-path $P=\left(\left(v_{i-1},v_i,t_i\right)\right)_{i=1}^k$ is \emph{fastest} if for all temporal $(s,z)$-path $P'=\left(\left(v'_{i-1},v'_i,t'_i\right)\right)_{i=1}^{k'}$ we have that $t_k-t_0\le t'_{k'}-t'_0$. We say that the \emph{duration} of $P$ is $d(P)=t_k-t_0+1$. Furthermore, the concept of \emph{travel delays} is very important for our proofs.

\begin{definition}[travel delays] \label{def:travelDelays}
	Let $(G, \lambda)$ be a $\Delta$-periodic temporal graph.
	Let $e_1=\{u,v\}$ and $e_2=\{v,w\}$ be two incident edges in $G$ with $e_1 \cap e_2 = \{v\}$.
	We define the \emph{travel delay} from~$u$ to $w$ at vertex $v$, denoted with $\tau_v^{u,w}$,
	as follows.
	\begin{equation*}
		\tau_v^{u,w} = \begin{cases}
        \lambda (e_2) - \lambda(e_1), & \text{if } \lambda (e_2) > \lambda(e_1),\\
        \lambda (e_2) - \lambda(e_1) + \Delta, & \text{otherwise}.
        \end{cases}
	\end{equation*}
\end{definition}
 From \Cref{def:travelDelays} we immediately get the following observation.
\begin{observation}\label{obs:traveldelays}
    Let $(G, \lambda)$ be a $\Delta$-periodic temporal graph.
	Let $e_1=\{u,v\}$ and $e_2=\{v,w\}$ be two incident edges in $G$ with $e_1 \cap e_2 = \{v\}$. Then we have that $\tau_v^{u,w} = \Delta - \tau_v^{w,u}$ if $\lambda (e_1) \neq \lambda(e_2)$, and $\tau_v^{u,w} = \tau_v^{w,u} =\Delta$ if $\lambda (e_1) = \lambda(e_2)$.
\end{observation}

Intuitively, the value of $\tau_v^{u,w}$ quantifies how long a fastest temporal path waits at vertex~$v$ when first traversing edge $\{u,v\}$ and then edge $\{v,w\}$. Formally, we have the following. 
\begin{lemma}\label{lem:duration}
    Let $P=\left(\left(v_{i-1},v_i,t_i\right)\right)_{i=1}^k$ be a fastest temporal $(s,z)$-path. Then we have
    $d(P)=1+\sum_{i\in[k-1]}\tau_{v_i}^{v_{i-1},v_{i+1}}$.
\end{lemma}
\begin{proof}
    This statement can straightforwardly be shown by induction on $k$. For $k=1$ it follows from \Cref{def:travelDelays} and the definition of the duration of a fastest temporal path. For~$k>1$ we obtain the following. 
    \begin{align*}
        d(P)&=t_k-t_0+1 = (t_{k-1}-t_0 +1) + t_k-t_{k-1}\\
        &= (1+\sum_{i\in[k-2]}\tau_{v_i}^{v_{i-1},v_{i+1}})+ t_k-t_{k-1}\\
        &= (1+\sum_{i\in[k-2]}\tau_{v_i}^{v_{i-1},v_{i+1}})+ \tau_{v_{k-1}}^{v_{k-2},v_{k}}\\
        &= 1+\sum_{i\in[k-1]}\tau_{v_i}^{v_{i-1},v_{i+1}}.
    \end{align*}
    Here, we use \Cref{def:travelDelays} and the fact that $P$ is a fastest temporal path to obtain $t_k-t_{k-1}=\tau_{v_{k-1}}^{v_{k-2},v_{k}}$.
\end{proof}

\section{The Classical Complexity on Restricted Instances}\label{sec:hardness}

In this section, we prove that \deltaUpperBound is NP-hard even in quite restrictive settings: when the input tree is a star and $\Delta\ge 3$, or when the input tree has constant diameter or constant maximum degree, and in addition $\Delta=2$. We remark that the diameter can be seen as a measure of how ``star-like'' a tree is. In particular, on trees, the diameter upper bounds the treedepth.

\begin{theorem}\label{thm:NPh}
    \deltaUpperBound\ is NP-hard even if the input matrix is symmetric and
    \begin{itemize}
    \item $\Delta\ge 3$ and $G$ is a star,
    \item $\Delta= 2$ and $G$ has constant diameter and a constant pathwidth, or
    \item $\Delta=2$ and $G$ has constant maximum degree and a constant pathwidth.
    \end{itemize} 
\end{theorem}

We start with showing the first statement in \Cref{thm:NPh} via a simple reduction from graph coloring. Essentially the same reduction was independently discovered by Meusel~et~al.~\cite{meusel2025directedtemporaltreerealization}.

\begin{proposition}\label{prop:star}
\deltaUpperBound is NP-hard even if $\Delta\ge3$ is constant and the input tree $G$ is a star.
\end{proposition}
\begin{proof}
We present a reduction from~\COL[$x$], which is known to be NP-hard for each~$x \geq 3$~\cite{Karp1972Reducibility}.

\problemdef{\COL[$x$]}{A graph $G = (V,E)$.}{Is there a~\emph{proper~$x$-coloring~$\chi$} of~$G$, that is, a function~$\chi \colon V \to [1,x]$, such that for each edge~$\{u,v\}\in E$, $\chi(u) \neq \chi(v)$?}

Let~$\Delta \geq 3$ and let~$G=(V,E)$ be an instance of~\COL[$\Delta$].
We obtain an equivalent instance~$I:=(G',D,\Delta)$ of~\deltaUpperBound as follows:
The graph~$G':=(V \cup \{c\}, E')$ is a star with center vertex~$c$ and leaves~$V$.
For each vertex~$v\in V$, we set $D_{v,c} = D_{c,v} = 1$.
Furthermore, for each two distinct vertices~$u$ and~$v$ of~$V$, we set
$$D_{u,v} := D_{v,u} := \begin{cases}\Delta & \text{if } \{u,v\}\in E \text{, and}\\ \Delta+1 & \text{ otherwise.} \end{cases}$$
This completes the construction of~$I$. The matrix $D$ is clearly symmetric.
Next, we show that~$G$ is~$\Delta$-colorable if and only if~$I$ is a yes-instance of~\deltaUpperBound.

$(\Rightarrow)$
Let~$\chi \colon V \to [1,\Delta]$ be a~$\Delta$-coloring of~$G$.
We define an edge labeling~$\lambda$ of~$G'$ that realized~$D$ as follows:
For each vertex~$v\in V$, we set~$\lambda(\{c,v\}) := \chi(v)$.
Next, we show that~$\lambda$ realizes~$D$.
To this end, note that all entries of~$D$ of value at most 1 are trivially realized.
Hence, it remains to show that for any two distinct vertices~$u$ and~$v$ of~$V$, the duration of any fastest temporal path from~$u$ to~$v$ (and vice versa) is at most~$D_{u,v}$.
Note that this trivially holds for all~$\{u,v\}\notin E$, since for such vertex pairs $D_{u,v} = D_{v,u} = \Delta + 1$, which is an upper bound for the duration of any temporal path of length 2.
Hence, it remains to consider the vertex pairs~$\{u,v\}$ that are edges of~$G$.
Since $\chi$ is a~$\Delta$-coloring of~$G$, $u$ and~$v$ receive distinct colors under~$\chi$ and, thus, the edges~$\{c,u\}$ and~$\{c,v\}$ receive distinct labels under~$\lambda$.
This implies that both temporal paths~$(u,c,v)$ and~$(v,c,u)$ have a duration of at most~$\Delta$ each.
Hence, $I$ is a yes-instance of~\deltaUpperBound, since $\lambda$ realizes~$D$.

$(\Leftarrow)$
Let~$\lambda\colon E' \to [1,\Delta]$ be an edge labeling of~$G'$ that realizes~$D$.
We define a $\Delta$-coloring~$\chi$ of the vertices of~$V$ as follows:
For each vertex~$v\in V$, we set~$\chi(v) := \lambda(\{c,v\})$.
Next, we show that for each edge~$\{u,v\}\in E$, $u$ and~$v$ receive distinct colors under~$\chi$.
Since~$\{u,v\}$ is an edge of~$E$, $D_{u,v} = \Delta$.
Hence, $\lambda(\{c,u\}) \neq \lambda(\{c,v\})$, as otherwise, the unique path~$(u,c,v)$ from~$u$ to~$v$ in~$G'$ has a duration of exactly~$\Delta +1 > D_{u,v}$.
This implies that~$\chi(u) = \lambda(\{c,u\}) \neq \lambda(\{c,v\}) = \chi(v)$.
Consequently, $G$ is a yes-instance of~\COL[$\Delta$].
\end{proof}

To show the second and third statement of \Cref{thm:NPh} we present two reductions from the NP-hard problem \textsc{Monotone Not-All-Equal 3-Satisfiability (MonNAE3SAT)} to \deltaUpperBound. It is known that \textsc{MonNAE3SAT} is NP-hard~\cite{Schaefer1978complexity}.

\problemdef{\textsc{Monotone Not-All-Equal 3-Satisfiability (MonNAE3SAT)}}{A set $X$ of Boolean variables and a set $C\subseteq X^3$ of clauses, each of which contains three variables.}{Is there an assignment to the variables $X$ such that for each clause in $C$ at least one variable is set to true and at least one variable is set to false?}

Given an instance of \textsc{MonNAE3SAT}, we call an assignment of Boolean values to the variables $X$ \emph{satisfying}, if each clause in $C$ is \emph{satisfied}, that is, it holds that at least one variable in the clause is set to true and at least one variable is set to false. 
We give one reduction showing that \deltaUpperBound\ is NP-hard even if the input tree~$G$ has constant diameter and we give a second reduction showing that \deltaUpperBound\ is NP-hard even if $G$ has constant maximum degree, and, in both cases, $\Delta=2$. Both reductions are based on similar ideas. Before we give formal descriptions, we present several properties of labelings that we will exploit. 

The first property establishes a relation between the duration of a fastest temporal path from $s$ to $t$ and the duration of a fastest temporal path from $t$ back to $s$.

\begin{lemma}
\label{lem:durationsum}
    Let $G$ be a tree and $(G, \lambda)$ be a $\Delta$-periodic temporal graph.
    Let $P_{s,t}=\left(\left(v_{i-1},v_i,t_i\right)\right)_{i=1}^k$ be a fastest temporal path in $(G, \lambda)$ from $s=v_0$ to $t=v_k$. Let $P_{t,s}$ be a fastest temporal path in $(G, \lambda)$ from $t$ to $s$.
	Then we have
 \begin{equation*}
     d(P_{s,t})+d(P_{t,s})\ge (k-1)\cdot\Delta+2.
 \end{equation*}
 Furthermore, if for all $i\in[k-1]$ we have $\lambda(\{v_{i-1},v_i\})\neq\lambda(\{v_i,v_{i+1}\})$, then we have
  \begin{equation*}
     d(P_{s,t})+d(P_{t,s})= (k-1)\cdot\Delta+2.
 \end{equation*}
\end{lemma}
\begin{proof}
 This follows directly from \Cref{lem:duration} and \Cref{obs:traveldelays}.
\end{proof}
From \Cref{lem:durationsum}, we immediately get the following.
\begin{corollary}\label{cor:forcing}
    Let $(G,D,\Delta)$ be a yes-instance of \deltaUpperBound. Let $P_{s,t}=\left(\left(v_{i-1},v_i,t_i\right)\right)_{i=1}^k$ with $k>1$ be a fastest temporal path in $(G, \lambda)$ from $s=v_0$ to $t=v_k$. Assume that $D_{s,t}+D_{t,s}=(k-1)\cdot\Delta+2$. Then for every solution $\lambda$, we have
    $\lambda(\{v_0,v_1\})=\lambda(\{v_{k-1},v_k\})-D_{s,t}+i\cdot\Delta+1$,
    for some (uniquely determined) $i\ge 0$.
\end{corollary}
Intuitively, \Cref{cor:forcing} allows us to use the matrix $D$ ``force'' certain labels of certain pairs of edges to be a function of each other. In particular, we can force the labels of certain pairs of edges to be the same. We will heavily exploit this in the gadgets of the reduction.

The next property provides some insight into the maximum duration that any fastest temporal path in a $\Delta$-periodic temporal graph can have.
\begin{lemma}
\label{lem:maxdur}
    Let $(G, \lambda)$ be a $\Delta$-periodic temporal graph.
    Let $P=\left(\left(v_{i-1},v_i,t_i\right)\right)_{i=1}^k$ be a fastest temporal path in $(G, \lambda)$.
	Then we have $d(P)\le (k-1)\cdot\Delta+1$.
 Furthermore, if~$d(P)= (k-1)\cdot\Delta+1$, then we have for all $i\in[k-1]$ that $\lambda(\{v_{i-1},v_i\})=\lambda(\{v_i,v_{i+1}\})$.
\end{lemma}
\begin{proof}
This follows directly from \Cref{lem:duration} and \Cref{obs:traveldelays}.
\end{proof}
Intuitively, \Cref{lem:maxdur} tells us which values in $D$ essentially do not put any constraints on a solution labeling. More formally, given a graph $G$ and two vertices $s,t$ of distance $k$ in $G$, then we say that the upper bound $D_{s,t}$ is \emph{trivial} if $D_{s,t}=(k-1)\cdot\Delta+1$. 

Next, we present variable gadgets and clause gadgets that we use in the reductions. The above properties will help us prove that the gadgets have the desired functionality.
From now on, assume that we are given an instance $(X,C)$ of \textsc{MonNAE3SAT}. We set $\Delta=2$.

For each variable $x\in X$ we create a \emph{variable gadget} $(G_x,D)$, where $G_x=(V_x,E_x)$ with $V_x=\{w_x,v_{x,1},v_{x,2}\}$ and $E_x=\{\{w_x,v_{x,1}\},\{w_x,v_{x,2}\}\}$. We set $D_{v_{x,1},v_{x,2}}=D_{v_{x,1},v_{x,2}}=2$ and make all other upper bounds trivial. For an illustration see \Cref{fig:gadgets}. Informally speaking, the upper bound enforces that the two edges of the variable gadgets have different labels.

For each clause $c\in C$ we create a \emph{clause gadget} $(G_c,D)$, where $G_c=(V_c,E_c)$. Let $c=(x,y,z)$. We define $V_c$ and $E_c$ as follows. 
\begin{itemize}
    \item $V_c=\{w_c,u_{c,1},u_{c,2},u_{c,3},u_{c,x,1},u_{c,x,2},u_{c,y,1},u_{c,y,2},u_{c,y,3},u_{c,y,4},u_{c,z,1},u_{c,z,2},$ $u_{c,z,3},u_{c,z,4},u_{c,z,5},u_{c,z,6}\}$, and
    \item $E_c=\{\{w_c,u_{c,1}\},\{u_{c,1},u_{c,2}\},\{u_{c,2},u_{c,3}\},\{w_c,u_{c,x,1}\},\{w_c,u_{c,x,2}\},\{w_c,u_{c,y,1}\},$ $\{w_c,u_{c,y,2}\},\{w_c,u_{c,z,1}\},\{w_c,u_{c,z,2}\},\{u_{c,1},u_{c,y,3}\},\{u_{c,1},u_{c,y,4}\},\{u_{c,1},u_{c,z,3}\},$ $\{u_{c,1},u_{c,z,4}\},\{u_{c,2},u_{c,z,5}\},\{u_{c,2},u_{c,z,6}\}\}$.
\end{itemize}
For an illustration see \Cref{fig:gadgets}. 
Furthermore, we set 
\vspace{-2ex}
\begin{multicols}{2}
\begin{itemize}
    \item $D_{u_{c,x,1},u_{c,x,2}}=D_{u_{c,x,2},u_{c,x,1}}=2$,
    \item $D_{u_{c,y,1},u_{c,y,2}}=D_{u_{c,y,2},u_{c,y,1}}=2$,
    \item $D_{u_{c,y,3},u_{c,y,4}}=D_{u_{c,y,4},u_{c,y,3}}=2$,
    \item $D_{u_{c,z,1},u_{c,z,2}}=D_{u_{c,z,2},u_{c,z,1}}=2$,
    \item $D_{u_{c,z,3},u_{c,z,4}}=D_{u_{c,z,4},u_{c,z,3}}=2$,
    \item $D_{u_{c,z,5},u_{c,z,6}}=D_{u_{c,z,6},u_{c,z,5}}=2$.
\end{itemize}
\end{multicols}
\vspace{-2ex}
\noindent We refer to the above as the \emph{first set of upper bounds}. We set 
\vspace{-2ex}
\begin{multicols}{2}
\begin{itemize}
    \item $D_{u_{c,x,2},u_{c,1}}=D_{u_{c,1},u_{c,x,2}}=2$,
\end{itemize}
\end{multicols}
\vspace{-2ex}
\noindent which we refer to as the \emph{second set of upper bounds}. Moreover, we set
\vspace{-2ex}
\begin{multicols}{2}
\begin{itemize}
    \item $D_{u_{c,y,1},u_{c,y,4}}=D_{u_{c,y,4},u_{c,y,1}}=4$,
    \item $D_{u_{c,y,2},u_{c,y,3}}=D_{u_{c,y,3},u_{c,y,2}}=4$,
    \item $D_{u_{c,y,4},u_{c,2}}=D_{u_{c,2},u_{c,y,4}}=2$,
    \item[\vspace{\fill}]
\end{itemize}
\end{multicols}
\vspace{-2ex}
\noindent We refer to the above as the \emph{third set of upper bounds}. Moreover, we set 
\vspace{-2ex}
\begin{multicols}{2}
\begin{itemize}
    \item $D_{u_{c,z,1},u_{c,z,4}}=D_{u_{c,z,4},u_{c,z,1}}=4$,
    \item $D_{u_{c,z,2},u_{c,z,3}}=D_{u_{c,z,3},u_{c,z,2}}=4$,
    \item $D_{u_{c,z,3},u_{c,z,6}}=D_{u_{c,z,6},u_{c,z,3}}=4$,
    \item $D_{u_{c,z,4},u_{c,z,5}}=D_{u_{c,z,5},u_{c,z,4}}=4$,
    \item $D_{u_{c,z,6},u_{c,3}}=D_{u_{c,3},u_{c,z,6}}=2$,
    \item[\vspace{\fill}]
\end{itemize}
\end{multicols}
\vspace{-2ex}
\noindent We refer to the above as the \emph{fourth set of upper bounds}. Finally, we set 
\vspace{-2ex}
\begin{multicols}{2}
\begin{itemize}
    \item $D_{w_{c},u_{c,3}}=D_{u_{c,3},w_{c}}=4$,
\end{itemize}
\end{multicols}
\vspace{-2ex}
\noindent which we refer to as the \emph{fifth set of upper bounds}. We make all other upper bounds trivial.

The intuition behind the clause gadget is as follows. For an illustration see \Cref{fig:gadgets}. The red edges in the figure are related to variable $x$, the green edges to variable $y$, and the blue edges to variable $z$. In particular, we can think of the subgraph induced on $w_c, u_{c,x,1},u_{c,x,2}$ as a copy of the variable gadget of~$x$. When we connect the gadgets, we will enforce that the label on the edge $\{w_c,u_{c,x,1}\}$ is the same as the label on the edge $\{w_x,v_{x,1}\}$ from the variable gadget for $x$. Similarly, we can think of the subgraphs induced on $w_c, u_{c,y,1},u_{c,y,2}$ and $u_{c,1}, u_{c,y,3},u_{c,y,4}$ as a copies of the variable gadget of $y$. We have an analogous situation for variable $z$, with three copies of the gadget. The first set of upper bounds ensures that the above-mentioned subgraphs actually behave in the same way as the variable gadgets they are supposed to copy.
The first two upper bounds in the second set of upper bounds copy the label of the edge $\{w_c,u_{c,x,1}\}$ (which encodes the truth value of variable $x$) to the edge $\{w_c,u_{c,1}\}$. The third set of upper bounds copies the labels of the subgraph induced on $w_c, u_{c,y,1},u_{c,y,2}$ to the subgraph induced on $u_{c,1}, u_{c,y,3},u_{c,y,4}$. Intuitively, it makes sure that those two copies of the variable gadget for $y$ have the same labeling. The last upper bound in the second set of upper bounds copies the label of the edge $\{u_{c,1},u_{c,y,3}\}$ (which encodes the truth value of variable $y$) to the edge $\{u_{c,1},u_{c,2}\}$. The fourth set of upper bounds plays an analogous role for variable $z$, and ultimately copies the label of the edge $\{u_{c,2},u_{c,z,5}\}$ (which encodes the truth value of variable $y$) to the edge $\{u_{c,2},u_{c,3}\}$.
Now we have that the labels on the edges $\{w_c,u_{c,1}\}$, $\{u_{c,1},u_{c,2}\}$, and 
$\{u_{c,2},u_{c,3}\}$ encode the truth values of variables $x$, $y$, and $z$, respectively. Furthermore, they form a path in the gadget. The fifth set of upper bounds enforces that the three edges in this path may not all have the same label. This corresponds to the clause being satisfied.

\begin{figure}[t]
\begin{center}
\ifarxiv
    \includegraphics[scale=1]{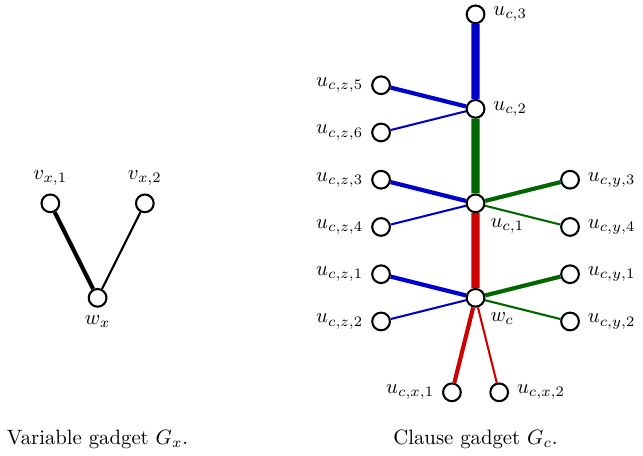}
\else
\begin{tikzpicture}[line width=1pt,scale=.8]
    \node[vert,label=below:$w_x$] (V1) at (2,0) {}; 
    \node[vert,label=above:$v_{x,1}$] (V2) at (1,2) {}; 
    \node[vert,label=above:$v_{x,2}$] (V3) at (3,2) {}; 

    \node (V4) at (2,-3) {Variable gadget $G_x$.};

    \draw[line width=2pt] (V1) -- (V2);
    \draw (V1) -- (V3);

    \node[vert,label=below right:$w_c$] (U1) at (10,0) {}; 
    \node[vert,label=left:$u_{c,z,2}$] (UU1) at (8,-.5) {}; 
    \node[vert,label=left:$u_{c,z,1}$] (UU2) at (8,.5) {}; 
    \node[vert,label=right:$u_{c,y,2}$] (UU3) at (12,-.5) {}; 
    \node[vert,label=right:$u_{c,y,1}$] (UU4) at (12,.5) {};
    \node[vert,label=left:$u_{c,x,1}$] (UU5) at (9.5,-2) {}; 
    \node[vert,label=right:$u_{c,x,2}$] (UU6) at (10.5,-2) {};
    
    \node[vert,label=below right:$u_{c,1}$] (U2) at (10,2) {}; 
    
    \node[vert,label=left:$u_{c,z,4}$] (U3) at (8,1.5) {}; 
    \node[vert,label=left:$u_{c,z,3}$] (U32) at (8,2.5) {}; 
    \node[vert,label=right:$u_{c,y,4}$] (U4) 
    at (12,1.5) {}; 
    \node[vert,label=right:$u_{c,y,3}$] (U42) 
    at (12,2.5) {}; 
    
    \node[vert,label=right:$u_{c,2}$] (U5) at (10,4) {}; 
    \node[vert,label=left:$u_{c,z,6}$] (U6) at (8,3.5) {};  
    \node[vert,label=left:$u_{c,z,5}$] (U62) at (8,4.5) {}; 
    \node[vert,label=right:$u_{c,3}$] (U7) at (10,6) {}; 

    \node (U8) at (10,-3) {Clause gadget $G_c$.};

    \draw[color=blue!80!black] (U1) -- (UU1);
    \draw[line width=2pt,color=blue!80!black] (U1) -- (UU2);
    \draw[color=green!40!black] (U1) -- (UU3);
    \draw[line width=2pt,color=green!40!black] (U1) -- (UU4);
    \draw[line width=2pt,color=red!80!black] (U1) -- (UU5);
    \draw[color=red!80!black] (U1) -- (UU6);
    \draw[line width=4pt,color=red!80!black] (U1) -- (U2);
    \draw[color=blue!80!black] (U2) -- (U3);
    \draw[line width=2pt,color=blue!80!black] (U2) -- (U32);
    \draw[color=green!40!black] (U2) -- (U4);
    \draw[line width=2pt,color=green!40!black] (U2) -- (U42);
    \draw[line width=4pt,color=green!40!black] (U2) -- (U5);
    \draw[color=blue!80!black] (U5) -- (U6);
    \draw[line width=2pt,color=blue!80!black] (U5) -- (U62);
    \draw[line width=4pt,color=blue!80!black] (U5) -- (U7);
\end{tikzpicture}
\fi
    \end{center}
    \caption{Visualization of the variable gadget for variable $x$ (left) and the clause gadget for clause $(x,y,z)$ (right). Informally, the label on the bold edge of the variable gadget models whether the variable is set to true or false. 
    In the clause gadget, the red edges are associated with variable $x$, the green ones with variable $y$, and the blue ones with variable $z$.
    The labels on the bold edges in the clause gadget are forced to have the same label as the bold edge in the variable gadget of the associated variable. In particular, this holds for the three very thick edges. Informally, the clause gadget forbids that all three of those edges obtain the same label, which corresponds to the clause being satisfied if not all three variables are set to the same truth value.}\label{fig:gadgets}
\end{figure}

For the two hardness results (the second and third statement in \Cref{thm:NPh}), we arrange the variable and clause gadgets in different ways. One way produces a tree with constant diameter, and a second way produces a tree with constant maximum degree. 

\medskip
\noindent\textbf{Reduction for trees of constant diameter.} 
We first formally describe how to create an instance $(G^{(1)},D^{(1)},\Delta)$ of \deltaUpperBound where $G^{(1)}$ is a tree with \textit{constant diameter}.
We start with a vertex $w^\star$. Now, for each variable $x$, we identify the vertex $w_x$ of the variable gadget~$G_x$ with~$w^\star$. For each clause $c$, we also identify the vertex $w_c$ of the clause gadget $G_c$ with~$w^\star$. This finishes the construction of $G^{(1)}$, for an illustration see \Cref{fig:hard1}.

\begin{figure}[t]
\begin{center}
\ifarxiv
    \includegraphics[scale=1]{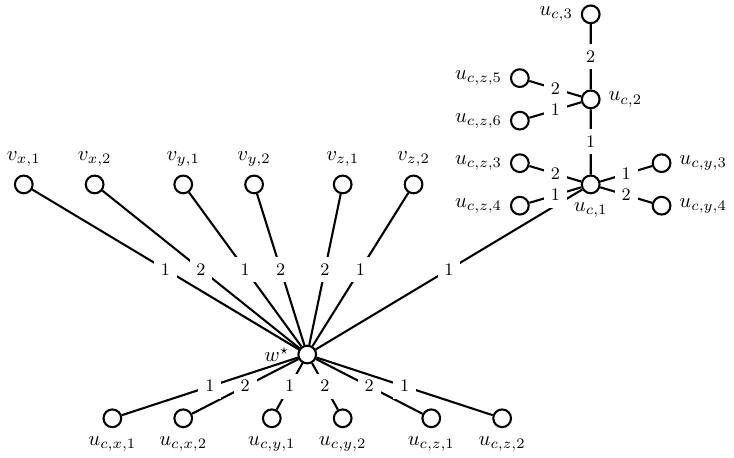}
    \else
\begin{tikzpicture}[line width=1pt,scale=.6,yscale=1.2]

    \node[vert,label=left:$w^\star$] (W) at (8,-2) {}; 
    
    \node[vert,label=above:$v_{x,1}$] (V2) at (0,2) {}; 
    \node[vert,label=above:$v_{x,2}$] (V3) at (2,2) {}; 
    
    \node[vert,label=above:$v_{y,1}$] (V22) at (4.5,2) {}; 
    \node[vert,label=above:$v_{y,2}$] (V23) at (6.5,2) {}; 

    \node[vert,label=above:$v_{z,1}$] (V32) at (9,2) {}; 
    \node[vert,label=above:$v_{z,2}$] (V33) at (11,2) {}; 



    \draw (W) --node[timelabel] {$1$} (V2);
    \draw (W) --node[timelabel] {$2$} (V3);
    \draw (W) --node[timelabel] {$1$} (V22);
    \draw (W) --node[timelabel] {$2$} (V23);
    \draw (W) --node[timelabel] {$2$} (V32);
    \draw (W) --node[timelabel] {$1$} (V33);

    \node[vert,label=below:$u_{c,1}$] (U2) at (16,2) {}; 
    \node[vert,label=left:$u_{c,z,4}$] (U3) at (14,1.5) {}; 
    \node[vert,label=left:$u_{c,z,3}$] (U32) at (14,2.5) {}; 
    \node[vert,label=right:$u_{c,y,4}$] (U4) at (18,1.5) {}; 
    \node[vert,label=right:$u_{c,y,3}$] (U42) at (18,2.5) {}; 
    \node[vert,label=right:$u_{c,2}$] (U5) at (16,4) {}; 
    \node[vert,label=left:$u_{c,z,6}$] (U6) at (14,3.5) {}; 
    \node[vert,label=left:$u_{c,z,5}$] (U62) at (14,4.5) {}; 
    \node[vert,label=left:$u_{c,3}$] (U7) at (16,6) {}; 
    \node[vert,label=below:$u_{c,x,1}$] (UU1) at (2.5,-3.5) {}; 
    \node[vert,label=below:$u_{c,x,2}$] (UU2) at (4.5,-3.5) {};
    \node[vert,label=below:$u_{c,y,1}$] (UU3) at (7,-3.5) {}; 
    \node[vert,label=below:$u_{c,y,2}$] (UU4) at (9,-3.5) {};
    \node[vert,label=below:$u_{c,z,1}$] (UU5) at (11.5,-3.5) {}; 
    \node[vert,label=below:$u_{c,z,2}$] (UU6) at (13.5,-3.5) {};


    \draw (W) --node[timelabel] {$1$} (U2);
    \draw (U2) --node[timelabel] {$1$} (U3);
    \draw (U2) --node[timelabel] {$2$} (U32);
    \draw (U2) --node[timelabel] {$2$} (U4);
    \draw (U2) --node[timelabel] {$1$} (U42);
    \draw (U2) --node[timelabel] {$1$} (U5);
    \draw (U5) --node[timelabel] {$1$} (U6);
    \draw (U5) --node[timelabel] {$2$} (U62);
    \draw (U5) --node[timelabel] {$2$} (U7);
    \draw (W) --node[timelabel] {$2$} (UU2);
    \draw (W) --node[timelabel] {$1$} (UU1);
    \draw (W) --node[timelabel] {$1$} (UU3);
    \draw (W) --node[timelabel] {$2$} (UU4);
    \draw (W) --node[timelabel] {$2$} (UU5);
    \draw (W) --node[timelabel] {$1$} (UU6);
    
\end{tikzpicture}
\fi
    \end{center}
    \caption{Illustration of $G^{(1)}$ with three variable gadgets for variables $x,y,z$, a clause gadget for clause $c=(x,y,z)$. A labeling $\lambda$ is illustrated that would be created (see proof of \Cref{lem:corrtd1}) if $x$ is set to true, $y$ is set to true, and $z$ is set to false.}\label{fig:hard1}
\end{figure}

Now we describe how to construct $D^{(1)}$. First of all, for each pair of vertices from the same gadget, the matrix $D^{(1)}$ adopts the values from the matrix of that gadget. 
Let $x$ be a variable that appears in clause $c$. Then we set $D^{(1)}_{v_{x,1},u_{c,x,2}}=D^{(1)}_{u_{c,x,2},v_{x,1}}=2$ and $D^{(1)}_{v_{x,2},u_{c,x,1}}=D^{(1)}_{u_{c,x,1},v_{x,2}}=2$.
We set all other entries of $D^{(1)}$ to trivial upper bounds.
This finishes the construction of the instance $(G^{(1)},D^{(1)},\Delta)$. 
Clearly, we have the following.
\begin{observation}\label{obs:red1}
    The instance $(G^{(1)},D^{(1)},\Delta)$ of \deltaUpperBound can be computed in polynomial time.
\end{observation}

Furthermore, the following is easy to observe.
\begin{observation}\label{obs:td}
    The tree $G^{(1)}$ has diameter $6$, constant pathwidth, and the matrix $D^{(1)}$ is symmetric.
\end{observation}

In the following we show that $(G^{(1)},D^{(1)},\Delta)$ is a yes-instance of \deltaUpperBound if and only if~$(X,C)$ is a yes-instance of \textsc{MonNAE3SAT}. We begin by proving the following direction of the statement.

\begin{lemma}\label{lem:corrtd1}
    If $(X,C)$ is a yes-instance of \textsc{MonNAE3SAT}, then $(G^{(1)},D^{(1)},\Delta)$ is a yes-instance of \deltaUpperBound.
\end{lemma}
\begin{proof}
Let $(X,C)$ be a yes-instance of \textsc{MonNAE3SAT} and let $\phi:X\rightarrow \{\text{true},\text{false}\}$ be a solution for $(X,C)$. Then we construct a solution $\lambda$ for $(G^{(1)},D^{(1)},\Delta)$ as follows. See \Cref{fig:hard1} for an illustration of the solution we create.

Consider the variable gadget for variable $x$. If $\phi(x)=\text{true}$, then we set $\lambda(\{w^\star,v_{x,1}\})=1$ and $\lambda(\{w^\star,v_{x,2}\})=2$. Otherwise, we set $\lambda(\{w^\star,v_{x,1}\})=2$ and $\lambda(\{w^\star,v_{x,2}\})=1$.
Note that in both cases the upper bound $D^{(1)}_{v_{x,1},v_{x,2}}=D^{(1)}_{v_{x,2},v_{x,1}}=2$ is respected (recall the we identified $w_x$ with $w^\star$). Now we have given labels to all edges corresponding to variable gadgets.

Next, consider clause $c=(x_1,x_2,x_3)$. Then for all $i\in\{1,2,3\}$ we set $\lambda(\{w^\star,u_{c,x_i,1}\})=\lambda(\{w^\star,v_{x_i,1}\})$ and $\lambda(\{w^\star,u_{c,x_i,2}\})=\lambda(\{w^\star,v_{x_i,2}\})$. It is straightforward to verify that this ensures that the upper bounds $D^{(1)}_{v_{x_i,1},u_{c,x_i,2}}=D^{(1)}_{u_{c,x_i,2},v_{x_i,1}}=2$ and $D^{(1)}_{v_{x_i,2},u_{c,x_i,1}}=D^{(1)}_{u_{c,x_i,1},v_{x_i,2}}=2$ are respected. 
Now for all  $i\in\{2,3\}$, we set $\lambda(\{u_{c,1},u_{c,x_i,3}\})=\lambda(\{w^\star,u_{c,x_i,1}\})$ and $\lambda(\{u_{c,1},u_{c,x_i,4}\})=\lambda(\{w^\star,u_{c,x_i,2}\})$. Furthermore, we set $\lambda(\{u_{c,2},u_{c,x_3,5}\})=\lambda(\{w^\star,u_{c,x_3,1}\})$ and $\lambda(\{u_{c,2},u_{c,x_3,6}\})=\lambda(\{w^\star,u_{c,x_3,2}\})$.
Finally, we set $\lambda(\{w^\star,u_{c,1}\})=\lambda(\{w^\star,u_{c,x_1,1}\})$, $\lambda(\{u_{c,1},u_{c,2}\})=\lambda(\{w^\star,u_{c,x_2,1}\})$, and $\lambda(\{u_{c,2},u_{c,3}\})=\lambda(\{w^\star,u_{c,x_3,1}\})$.

It is straightforward to observe that the first set of upper bounds of the clause gadgets is respected. The second set of upper bounds is respected since we have set $\lambda(\{w^\star,u_{c,1}\})=\lambda(\{w^\star,u_{c,x_1,1}\})$ and we have $\lambda(\{w^\star,u_{c,x_1,1}\})\neq\lambda(\{w^\star,u_{c,x_1,2}\})$.
The third set of upper bounds is respected since we have the following. Note that we have set $\lambda(\{u_{c,1},u_{c,x_2,3}\})=\lambda(\{w^\star,u_{c,x_2,1}\})$ and $\lambda(\{u_{c,1},u_{c,x_2,4}\})=\lambda(\{w^\star,u_{c,x_2,2}\})$ and that $\lambda(\{u_{c,1},u_{c,x_2,3}\})\neq\lambda(\{u_{c,1},u_{c,x_2,4}\})$. Hence, independently from which label we set on edge $\{w^\star,u_{c,1}\}$, we have that the paths from $u_{c,x_2,1}$ to $u_{c,x_2,4}$ and $u_{c,x_2,2}$ to $u_{c,x_2,3}$, respectively, have three edges that do not all have the same label. It follows that the upper bounds $D^{(1)}_{u_{c,x_2,1},u_{c,x_2,4}}=D^{(1)}_{u_{c,x_2,4},u_{c,x_2,1}}=4$ and $D^{(1)}_{u_{c,x_2,2},u_{c,x_2,3}}=D^{(1)}_{u_{c,x_2,3},u_{c,x_2,2}}=4$ are respected. Finally, we have set $\lambda(\{u_{c,1},u_{c,2}\})=\lambda(\{w^\star,u_{c,x_2,1}\})$ and hence we have that $\lambda(\{u_{c,1},u_{c,2}\})\neq\lambda(\{u_{c,1},u_{c,x_2,4}\})$. It follows that the upper bounds $D^{(1)}_{u_{c,x_2,4},u_{c,2}}=D^{(1)}_{u_{c,2},u_{c,x_2,4}}=2$ is respected. 
By an analogous argument, we can also show that the fourth set of upper bounds is respected. 
Lastly, assume for contradiction that the fifth set of upper bounds is not respected. Then we have that the upper bound $D^{(1)}_{w^\star,u_{c,3}}=D^{(1)}_{u_{c,3},w^\star}=4$ is not respected. Note that then the temporal path from $w^\star$ to $u_{c,3}$ must have a duration of at least 5. This implies that we must have $\lambda(\{w^\star,u_{c,1}\})=\lambda(\{u_{c,1},u_{c,2}\})=\lambda(\{u_{c,2},u_{c,3}\})$. However, by the construction of~$\lambda$, this implies that $\phi(x_1)=\phi(x_2)=\phi(x_3)$, a contradiction to the assumption that clause $c$ is satisfied.

 Hence, since all other upper bounds in $D^{(1)}$ are trivial, we can conclude that the constructed labeling $\lambda$ is indeed a solution for $(G^{(1)},D^{(1)},\Delta)$.
\end{proof}

We proceed with showing the other direction of the correctness.

\begin{lemma}\label{lem:corrtd2}
    If $(G^{(1)},D^{(1)},\Delta)$ is a yes-instance of \deltaUpperBound, then $(X,C)$ is a yes-instance of \textsc{MonNAE3SAT}.
\end{lemma}
\begin{proof}
Let $(G^{(1)},D^{(1)},\Delta)$ be a yes-instance of \deltaUpperBound and let labeling $\lambda$ be a solution for the instance. We construct a satisfying assignment for the variables in $X$ as follows.

For each variable $x$, if $\lambda(\{w^\star,v_{x,1}\})=1$, then we set the variable $x$ to true. Otherwise, we set the variable $x$ to false. We claim that this yields a satisfying assignment for $(X,C)$.

Assume for contradiction that it is not. Then there is a clause $c\in C$ that is not satisfied, that is, either all three variables in the clause are set to true or all three variables are set to false. Assume that $c=(x_1,x_2,x_3)$ and that all variables in $c$ are set to true (the case where all of them are set to false is symmetric). This means that $\lambda(\{w^\star,v_{x_1,1}\})=\lambda(\{w^\star,v_{x_2,1}\})=\lambda(\{w^\star,v_{x_3,1}\})=1$. 
Note that for all $\ell\in\{1,2,3\}$ we have that $D^{(1)}_{v_{x_\ell,1},v_{x_\ell,2}}=D^{(1)}_{v_{x_\ell,1},v_{x_\ell,2}}=2$. It follows by \Cref{cor:forcing} that we must have $\lambda(\{w^\star,v_{x_1,2}\})=\lambda(\{w^\star,v_{x_2,2}\})=\lambda(\{w^\star,v_{x_3,2}\})=2$. 
Furthermore, for all $\ell\in\{1,2,3\}$ we have $D^{(1)}_{v_{x_\ell,1},u_{c,x_\ell,2}}=D^{(1)}_{u_{c,x_\ell,2},v_{x_\ell,1}}=2$ and $D^{(1)}_{v_{x_\ell,2},u_{c,x_\ell,1}}=D^{(1)}_{u_{c,x_\ell,1},v_{x_\ell,2}}=2$. 
By \Cref{cor:forcing}, this means that we must have $\lambda(\{w^\star,u_{c,x_1,1}\})=\lambda(\{w^\star,u_{c,x_2,1}\})=\lambda(\{w^\star,u_{c,x_3,1}\})=1$ and $\lambda(\{w^\star,u_{c,x_1,2}\})=\lambda(\{w^\star,u_{c,x_2,2}\})=\lambda(\{w^\star,u_{c,x_3,2}\})=2$.
Since we have $D^{(1)}_{u_{c,x_1,2},u_{c,1}}=D^{(1)}_{u_{c,1},u_{c,x_1,2}}=2$, by \Cref{cor:forcing} we have that $\lambda(\{w^\star,u_{c,1}\})=1$.
Now consider the third set of upper bounds for the clause gadget corresponding to $c$. We have that $D^{(1)}_{u_{c,x_2,1},u_{c,x_2,4}}=D^{(1)}_{u_{c,x_2,4},u_{c,x_2,1}}=4$. Since $\lambda(\{w^\star,u_{c,x_2,1}\})=\lambda(\{w^\star,u_{c,1}\})=1$, we must have that $\lambda(\{u_{c,1},u_{c,x_2,4}\})=2$ to respect the above upper bound. 
Furthermore, we have that $D^{(1)}_{u_{c,x_2,4},u_{c,2}}=D^{(1)}_{u_{c,2},u_{c,x_2,4}}=2$. By \Cref{cor:forcing} we must have that $\lambda(\{u_{c,1},u_{c,2}\})=1$.

By an analogous argument for the labels on edges in the clause gadget corresponding to $c$ that are associated with $x_3$, we get that $\lambda(\{u_{c,2},u_{c,3}\})=1$.
Summarizing, we get that 
\begin{equation*}
\lambda(\{w^\star,u_{c,1}\})=\lambda(\{u_{c,1},u_{c,2}\})=\lambda(\{u_{c,2},u_{c,3}\})=1.
\end{equation*}
This implies that the duration of a fastest path from $w^\star$ to $u_{c,3}$ in $(G^{(1)},\lambda)$ is $5$. However, we have $D^{(1)}_{w^\star,u_{c,3}}=4$ (recall that we identify $w_c$ with $w^\star$), a contradiction to the assumption that $\lambda$ is a solution for $(G^{(1)},D^{(1)},\Delta)$.
\end{proof}

\Cref{lem:corrtd1,lem:corrtd2} together with \Cref{obs:td,obs:red1} prove the second statement of \Cref{thm:NPh}, namely that \deltaUpperBound is NP-hard even if the input tree $G$ has constant diameter and $\Delta=2$. In the remainder, we prove that \deltaUpperBound is NP-hard even if the input tree $G$ has constant maximum degree and $\Delta=2$. We provide a second reduction from \textsc{MonNAE3SAT} that uses the same variable and clause gadgets, but arranges them in a different way. Intuitively, we replace the vertex~$w^\star$ in~$G^{(1)}$ with a long path.

\medskip
\noindent\textbf{Reduction for trees of constant maximum degree.} 
We now formally describe how to create an instance $(G^{(2)},D^{(2)},\Delta)$ of \deltaUpperBound where $G^{(2)}$ is a tree with \textit{constant maximum degree}.
For each variable $x$, we introduce a new vertex $w'_x$  and connect it to the vertex $w_x$ of the variable gadget~$G_x$, that is, we add edge $\{w_x,w'_x\}$. For each clause $c$, we introduce a new vertex $w'_c$ and connect it to the vertex $w_c$ of the clause gadget $G_c$, that is, we add edge $\{w_c,w'_c\}$. 
Next, we order the variables in some fixed, but arbitrary, way and we order the clauses in some fixed, but arbitrary, way.
Let $x$ and $y$ be two variables that are consecutive in the ordering, where $y$ is ordered after $x$. Then we connect $w_x$ with $w'_y$, that is, we add edge $\{w_x,w'_y\}$.
Let $c$ and $c'$ be two clauses that are consecutive in the ordering, where $c'$ is ordered after $c$. Then we connect $w_c$ with $w'_{c'}$, that is, we add edge $\{w_c,w'_{c'}\}$.
Finally, let $z$ be the variable that is ordered last and let $c''$ be the clause that is ordered first. Then we connect $w_z$ with $w'_{c''}$, that is, we add edge $\{w_z,w'_{c''}\}$.
This finishes the construction of $G^{(2)}$, for an illustration see \Cref{fig:hard2}. 

\begin{figure}[t]
\begin{center}
\ifarxiv
    \includegraphics[scale=1]{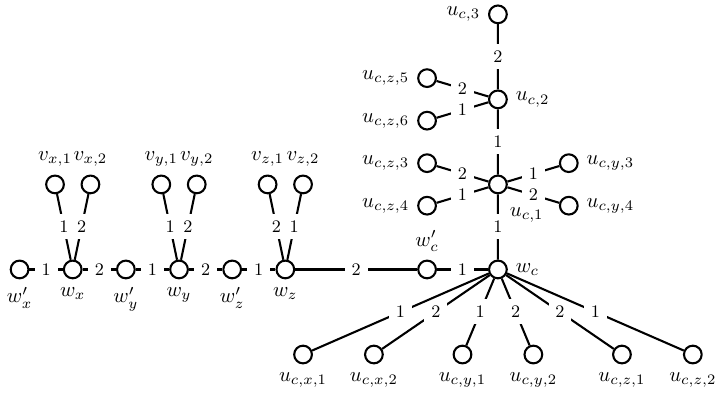}
    \else
\begin{tikzpicture}[line width=1pt,scale=.6,yscale=1.2] 
    
    \node[vert,label=below:$w'_x$] (V0) at (-1.5,0) {}; 
    \node[vert,label=below:$w_x$] (V1) at (0,0) {}; 
    \node[vert,label=above:$v_{x,1}$] (V2) at (-.5,2) {}; 
    \node[vert,label=above:$v_{x,2}$] (V3) at (.5,2) {}; 
    
    \node[vert,label=below:$w'_y$] (V20) at (1.5,0) {}; 
    \node[vert,label=below:$w_y$] (V21) at (3,0) {}; 
    \node[vert,label=above:$v_{y,1}$] (V22) at (2.5,2) {}; 
    \node[vert,label=above:$v_{y,2}$] (V23) at (3.5,2) {}; 

    \node[vert,label=below:$w'_z$] (V30) at (4.5,0) {}; 
    \node[vert,label=below:$w_z$] (V31) at (6,0) {}; 
    \node[vert,label=above:$v_{z,1}$] (V32) at (5.5,2) {}; 
    \node[vert,label=above:$v_{z,2}$] (V33) at (6.5,2) {}; 


    \draw (V0) --node[timelabel] {$1$} (V1);
    \draw (V1) --node[timelabel] {$1$} (V2);
    \draw (V1) --node[timelabel] {$2$} (V3);
    \draw (V1) --node[timelabel] {$2$} (V20);
    \draw (V20) --node[timelabel] {$1$} (V21);
    \draw (V21) --node[timelabel] {$1$} (V22);
    \draw (V21) --node[timelabel] {$2$} (V23);
    \draw (V21) --node[timelabel] {$2$} (V30);
    \draw (V30) --node[timelabel] {$1$} (V31);
    \draw (V31) --node[timelabel] {$2$} (V32);
    \draw (V31) --node[timelabel] {$1$} (V33);

    \node[vert,label=above:$w'_c$] (U0) at (10,0) {};
    \node[vert,label=right:$w_c$] (U1) at (12,0) {}; 
    \node[vert,label={[label distance=.1cm]-75:$u_{c,1}$}] (U2) at (12,2) {}; 
    \node[vert,label=left:$u_{c,z,4}$] (U3) at (10,1.5) {}; 
    \node[vert,label=left:$u_{c,z,3}$] (U32) at (10,2.5) {}; 
    \node[vert,label=right:$u_{c,y,4}$] (U4) at (14,1.5) {}; 
    \node[vert,label=right:$u_{c,y,3}$] (U42) at (14,2.5) {}; 
    \node[vert,label=right:$u_{c,2}$] (U5) at (12,4) {}; 
    \node[vert,label=left:$u_{c,z,6}$] (U6) at (10,3.5) {}; 
    \node[vert,label=left:$u_{c,z,5}$] (U62) at (10,4.5) {}; 
    \node[vert,label=left:$u_{c,3}$] (U7) at (12,6) {}; 
    \node[vert,label=below:$u_{c,x,1}$] (UU1) at (6.5,-2) {}; 
    \node[vert,label=below:$u_{c,x,2}$] (UU2) at (8.5,-2) {};
    \node[vert,label=below:$u_{c,y,1}$] (UU3) at (11,-2) {}; 
    \node[vert,label=below:$u_{c,y,2}$] (UU4) at (13,-2) {};
    \node[vert,label=below:$u_{c,z,1}$] (UU5) at (15.5,-2) {}; 
    \node[vert,label=below:$u_{c,z,2}$] (UU6) at (17.5,-2) {}; 
    


    \draw (V31) --node[timelabel] {$2$} (U0);
    \draw (U0) --node[timelabel] {$1$} (U1);
    \draw (U1) --node[timelabel] {$1$} (U2);
    \draw (U2) --node[timelabel] {$1$} (U3);
    \draw (U2) --node[timelabel] {$2$} (U32);
    \draw (U2) --node[timelabel] {$2$} (U4);
    \draw (U2) --node[timelabel] {$1$} (U42);
    \draw (U2) --node[timelabel] {$1$} (U5);
    \draw (U5) --node[timelabel] {$1$} (U6);
    \draw (U5) --node[timelabel] {$2$} (U62);
    \draw (U5) --node[timelabel] {$2$} (U7);
    \draw (U1) --node[timelabel] {$2$} (UU2);
    \draw (U1) --node[timelabel] {$1$} (UU1);
    \draw (U1) --node[timelabel] {$1$} (UU3);
    \draw (U1) --node[timelabel] {$2$} (UU4);
    \draw (U1) --node[timelabel] {$2$} (UU5);
    \draw (U1) --node[timelabel] {$1$} (UU6);
\end{tikzpicture}
\fi
    \end{center}
    \caption{Illustration of $G^{(2)}$. Depicted are three variable gadgets for variables $x,y,z$ where $x$ is ordered first, $y$ is ordered second, and $z$ is ordered third. Additionally, a clause gadget for clause~$c=(x,y,c)$ is depicted. A labeling $\lambda$ is illustrated that would be created (see proof of \Cref{lem:corrmd1}) if $x$ is set to true, $y$ is set to true, and $z$ is set to false.}\label{fig:hard2}
\end{figure}

Now we describe how to construct $D^{(2)}$. First of all, for each pair of vertices from the same gadget, the matrix $D^{(2)}$ adopts the values from the matrix of that gadget. 
For each pair of variables $x,y$ that are ordered consecutively, we set $D^{(2)}_{w_x,w_y}=D^{(2)}_{w_y,w_x}=2$ and $D^{(2)}_{w'_x,w'_y}=D^{(2)}_{w'_y,w'_x}=2$.
For each pair of variables $c,c'$ that are ordered consecutively, we set $D^{(2)}_{w_{c},w_{c'}}=D^{(2)}_{w_{c'},w_c}=2$ and $D^{(2)}_{w'_{c},w'_{c'}}=D^{(2)}_{w'_{c'},w'_c}=2$.
Finally, let $z$ be the variable that is ordered last and let $c''$ be the clause that is ordered first. We set $D^{(2)}_{w_{z},w_{c''}}=D^{(2)}_{w_{c''},w_z}=2$ and $D^{(2)}_{w'_{z},w'_{c''}}=D^{(2)}_{w'_{c''},w'_z}=2$.

Let $c=(x_1,x_2,x_3)$. For $\ell\in\{1,2,3\}$, let $x_\ell$ be the $i_\ell$th variable in the ordering and let $c$ be the $j$th clause in the ordering. Denote $|X|=n$. 
Then we set the following values in~$D^{(2)}$. For each $\ell\in\{1,2,3\}$ we set 
\[
D^{(2)}_{v_{x_\ell,1},u_{c,x_\ell,1}}=D^{(2)}_{u_{c,x_\ell,1},v_{x_\ell,1}}= 2(n-i_\ell+j)+3,
\]
and
\[
D^{(2)}_{v_{x_\ell,2},u_{c,x_\ell,2}}=D^{(2)}_{u_{c,x_\ell,2},v_{x_\ell,2}}= 2(n-i_\ell+j)+3.
\]
We set all other entries of $D^{(2)}$ to trivial upper bounds.

This finishes the construction of the instance $(G^{(2)},D^{(2)},\Delta)$. Clearly, we have the following.
\begin{observation}\label{obs:red2}
    The instance $(G^{(2)},D^{(2)},\Delta)$ of \deltaUpperBound can be computed in polynomial time.
\end{observation}

Furthermore, the following is easy to observe.
\begin{observation}\label{obs:md}
    The tree $G^{(2)}$ has maximum degree $8$, constant pathwidth, and the matrix $D^{(2)}$ is symmetric.
\end{observation}

In the following we show that $(G^{(2)},D^{(2)},\Delta)$ is a yes-instance of \deltaUpperBound if and only if~$(X,C)$ is a yes-instance of \textsc{MonNAE3SAT}. This is done in a similar way as in \Cref{lem:corrtd1,lem:corrtd2}. We begin by proving the following direction of the statement.

\begin{lemma}\label{lem:corrmd1}
    If $(X,C)$ is a yes-instance of \textsc{MonNAE3SAT}, then $(G^{(2)},D^{(2)},\Delta)$ is a yes-instance of \deltaUpperBound.
\end{lemma}
\begin{proof}
Let $(X,C)$ be a yes-instance of \textsc{MonNAE3SAT} and let $\phi:X\rightarrow \{\text{true},\text{false}\}$ be a solution for $(X,C)$. Then we construct a solution $\lambda$ for $(G^{(2)},D^{(2)},\Delta)$ as follows. See \Cref{fig:hard2} for an illustration of the solution we create.

Consider the variable gadget for variable $x$. We set $\lambda(\{w'_x,w_x\})=1$. If $\phi(x)=\text{true}$, when we set $\lambda(\{w_x,v_{x,1}\})=1$ and $\lambda(\{w_x,v_{x,2}\})=2$. Otherwise, we set $\lambda(\{w_x,v_{x,1}\})=1$ and $\lambda(\{w_x,v_{x,2}\})=2$. Note that in both cases the upper bound $D^{(2)}_{v_{x,1},v_{x,2}}=D^{(2)}_{v_{x,2},v_{x,1}}=2$ is respected. Lastly, if $x$ is not the last variable in the ordering, we set $\lambda(\{w_x,w'_y\})=2$, where $y$ is the variable directly after $x$ in the ordering, otherwise, we set $\lambda(\{w_x,w'_c\})=2$, where $c$ is the first clause in the ordering. 
Now we have given labels to all edges corresponding to variable gadgets. It is easy to check that for each pair of variables $x,y$ that are ordered consecutively, the upper bounds $D^{(2)}_{w_x,w_y}=D^{(2)}_{w_y,w_x}=2$ and $D^{(2)}_{w'_x,w'_y}=D^{(2)}_{w'_y,w'_x}=2$ are respected. Furthermore, the upper bounds $D^{(2)}_{w_{z},w_{c}}=D^{(2)}_{w_{c},w_z}=2$ and $D^{(2)}_{w'_{z},w'_{c}}=D^{(2)}_{w'_{c},w'_z}=2$ are respected, there $z$ is the last variable in the ordering and $c$ is the first clause in the ordering.

Next, consider clause $c=(x_1,x_2,x_3)$, and let $c$ be the $j$th clause in the ordering. 
We set $\lambda(\{w'_c,w_c\})=1$ and $\lambda(\{w_c,w'_{c'}\})=2$, where $c'$ is the $(j+1)$st clause (if it exists). It is easy to check that for each pair of clauses $c,c'$ that are ordered consecutively, the upper bounds $D^{(2)}_{w_c,w_{c'}}=D^{(2)}_{w_{c'},w_c}=2$ and $D^{(2)}_{w'_c,w'_{c'}}=D^{(2)}_{w'_{c'},w'_c}=2$ are respected.
Furthermore, for all $\ell\in\{1,2,3\}$ we set $\lambda(\{w_c,u_{c,x_\ell,1}\})=\lambda(\{w_{x_\ell},v_{x_\ell,1}\})$ and $\lambda(\{w_c,u_{c,x_\ell,2}\})=\lambda(\{w_{x_\ell},v_{x_\ell,2}\})$. 
For each $\ell\in\{1,2,3\}$, let $x_\ell$ be the $i_\ell$th variable in the ordering. 
Now we argue that for each $\ell\in\{1,2,3\}$, the upper bounds 
\[
D^{(2)}_{v_{x_\ell,1},u_{c,x_\ell,1}}=D^{(2)}_{u_{c,x_\ell,1},v_{x_\ell,1}}= 2(n-i_\ell+j)+3
\]
and
\[
D^{(2)}_{v_{x_\ell,2},u_{c,x_\ell,2}}=D^{(2)}_{u_{c,x_\ell,2},v_{x_\ell,2}}= 2(n-i_\ell+j)+3
\]
are respected.
Note that the distance from $w_{x_\ell}$ to $w_c$ in $G^{(2)}$ is $2(n-i_\ell+j)$. Furthermore, we have that along this path, all edges are labeled alternatingly with 1 and 2. It follows that the duration of a temporal path from $w_{x_\ell}$ to $w_c$ is $2(n-i_\ell+j)$. By construction, the first edge of this path is labeled with 2 and the last edge is labeled with 1.
Additionally, recall that by construction $\lambda(\{w_c,u_{c,x_\ell,1}\})=\lambda(\{w_{x_\ell},v_{x_\ell,1}\})$ and $\lambda(\{w_c,u_{c,x_\ell,2}\})=\lambda(\{w_{x_\ell},v_{x_\ell,2}\})$. 
If follows that both the temporal paths from $v_{x_\ell,1}$ to $u_{c,x_\ell,1}$ and from $v_{x_\ell,2}$ to $u_{c,x_\ell,2}$ have at most one pair of consecutive edges that are labeled with the same label. Furthermore, the temporal paths have length $2(n-i_\ell+j)+2$. It follows that their duration is at most $2(n-i_\ell+j)+3$ and the upper bounds are respected.

From now on the argument is essentially the same as in the proof of \Cref{lem:corrtd1}.
Now for all  $i\in\{2,3\}$, we set $\lambda(\{u_{c,1},u_{c,x_i,3}\})=\lambda(\{w_c,u_{c,x_i,1}\})$ and $\lambda(\{u_{c,1},u_{c,x_i,4}\})=\lambda(\{w_c,u_{c,x_i,2}\})$. Furthermore, we set $\lambda(\{u_{c,2},u_{c,x_3,5}\})=\lambda(\{w_c,u_{c,x_3,1}\})$ and $\lambda(\{u_{c,2},u_{c,x_3,6}\})=\lambda(\{w_c,u_{c,x_3,2}\})$.
Finally, we set $\lambda(\{w_c,u_{c,1}\})=\lambda(\{w_c,u_{c,x_1,1}\})$, $\lambda(\{u_{c,1},u_{c,2}\})=\lambda(\{w_c,u_{c,x_2,1}\})$, and $\lambda(\{u_{c,2},u_{c,3}\})=\lambda(\{w_c,u_{c,x_3,1}\})$.

It is straightforward to observe that the first set of upper bounds of the clause gadgets is respected. The second set of upper bounds is respected since we have set $\lambda(\{w_c,u_{c,1}\})=\lambda(\{w_c,u_{c,x_1,1}\})$ and we have $\lambda(\{w_c,u_{c,x_1,1}\})\neq\lambda(\{w_c,u_{c,x_1,2}\})$.
The third set of upper bounds is respected since we have the following. Note that we have set $\lambda(\{u_{c,1},u_{c,x_2,3}\})=\lambda(\{w_c,u_{c,x_2,1}\})$ and $\lambda(\{u_{c,1},u_{c,x_2,4}\})=\lambda(\{w_c,u_{c,x_2,2}\})$ and that $\lambda(\{u_{c,1},u_{c,x_2,3}\})\neq\lambda(\{u_{c,1},u_{c,x_2,4}\})$. Hence, independently from which label we set on edge $\{w_c,u_{c,1}\}$, we have that the paths from $u_{c,x_2,1}$ to $u_{c,x_2,4}$ and $u_{c,x_2,2}$ to $u_{c,x_2,3}$, respectively, have three edges that do not all have the same label. It follows that the upper bounds $D^{(2)}_{u_{c,x_2,1},u_{c,x_2,4}}=D^{(2)}_{u_{c,x_2,4},u_{c,x_2,1}}=4$ and $D^{(2)}_{u_{c,x_2,2},u_{c,x_2,3}}=D^{(2)}_{u_{c,x_2,3},u_{c,x_2,2}}=4$ are respected. Finally, we have set $\lambda(\{u_{c,1},u_{c,2}\})=\lambda(\{w_c,u_{c,x_2,1}\})$ and hence we have that $\lambda(\{u_{c,1},u_{c,2}\})\neq\lambda(\{u_{c,1},u_{c,x_2,4}\})$. It follows that the upper bounds $D^{(2)}_{u_{c,x_2,4},u_{c,2}}=D^{(2)}_{u_{c,2},u_{c,x_2,4}}=2$ is respected. 
By an analogous argument, we can also show that the fourth set of upper bounds is respected. 
Lastly, assume for contradiction that the fifth set of upper bounds is not respected. Then we have that the upper bound $D^{(2)}_{w_c,u_{c,3}}=D^{(2)}_{u_{c,3},w_c}=4$ is not respected. Note that then the temporal path from $w_c$ to $u_{c,3}$ must have a duration of at least 5. This implies that we must have $\lambda(\{w_c,u_{c,1}\})=\lambda(\{u_{c,1},u_{c,2}\})=\lambda(\{u_{c,2},u_{c,3}\})$. However, by the construction of~$\lambda$, this implies that $\phi(x_1)=\phi(x_2)=\phi(x_3)$, a contradiction to the assumption that clause $c$ is satisfied.

 Hence, since all other upper bounds in $D^{(2)}$ are trivial, we can conclude that the constructed labeling $\lambda$ is indeed a solution for $(G^{(2)},D^{(2)},\Delta)$.
\end{proof}

We proceed with showing the other direction of the correctness.

\begin{lemma}
\label{lem:corrmd2}
    If $(G^{(2)},D^{(2)},\Delta)$ is a yes-instance of \deltaUpperBound, then $(X,C)$ is a yes-instance of \textsc{MonNAE3SAT}.
\end{lemma}
\begin{proof}
Let $(G^{(2)},D^{(2)},\Delta)$ be a yes-instance of \deltaUpperBound and let labeling $\lambda$ be a solution for the instance. We construct a satisfying assignment for the variables in $X$ as follows.

For each variable $x$, if $\lambda(\{w_c,v_{x,1}\})=1$, then we set the variable $x$ to true. Otherwise, we set the variable $x$ to false. We claim that this yields a satisfying assignment for $(X,C)$.

Assume for contradiction that it is not. Then there is a clause $c\in C$ that is not satisfied, that is, either all three variables in the clause are set to true or all three variables are set to false. Assume that $c=(x_1,x_2,x_3)$ and that all variables in $c$ are set to true (the case where all of them are set to false is symmetric). This means that $\lambda(\{w_{x_1},v_{x_1,1}\})=\lambda(\{w_{x_2},v_{x_2,1}\})=\lambda(\{w_{x_3},v_{x_3,1}\})=1$. 
Note that for all $\ell\in\{1,2,3\}$ we have $D^{(2)}_{v_{x_\ell,1},v_{x_\ell,2}}=D^{(2)}_{v_{x_\ell,1},v_{x_\ell,2}}=2$. It follows by \Cref{cor:forcing} that we must have $\lambda(\{w_{x_1},v_{x_1,2}\})=\lambda(\{w_{x_2},v_{x_2,2}\})=\lambda(\{w_{x_2},v_{x_3,2}\})=2$. 

Furthermore, we can observe the following. Assume that $\lambda(\{w'_x,w_x\})=1$ (the case where $\lambda(\{w'_x,w_x\})=2$ is very similar, we discuss it afterwards). Then the upper bound $D^{(2)}_{w'_x,w'_y}=D^{(2)}_{w'_y,w'_x}=2$, where $y$ is the second variable in the ordering, enforces that $\lambda(\{w_x,w'_y\})=2$. The upper bound $D^{(2)}_{w_x,w_y}=D^{(2)}_{w_y,w_x}=2$ enforces that $\lambda(\{w'_y,w'_y\})=1$.
By iterating this argument, we obtain that for all variables $x$, we have that $\lambda(\{w'_x,w_x\})=1$ and $\lambda(\{w_x,w'_y\})=2$, where $y$ is the variable directly after $x$ in the ordering. Similarly, we get that $\lambda(\{w_x,w'_c\})=2$, if $x$ is the last variable in the ordering and $c$ is the first clause in the ordering. Finally, for each clause $c$ we get that $\lambda(\{w'_c,w_c\})=1$ and $\lambda(\{w_c,w'_{c'}\})=2$, where $c'$ is the clause directly after $c$ in the ordering.
Furthermore, for all $\ell\in\{1,2,3\}$ we have 
\[
D^{(2)}_{v_{x_\ell,1},u_{c,x_\ell,1}}=D^{(2)}_{u_{c,x_\ell,1},v_{x_\ell,1}}= 2(n-i_\ell+j)+3,
\]
and
\[
D^{(2)}_{v_{x_\ell,2},u_{c,x_\ell,2}}=D^{(2)}_{u_{c,x_\ell,2},v_{x_\ell,2}}= 2(n-i_\ell+j)+3.
\]
Note that the distance from $w_{x_\ell}$ to $w_c$ in $G^{(2)}$ is $2(n-i_\ell+j)$. Furthermore, by the arguments above, we have that along this path, all edges are labeled alternatingly with 1 and 2. It follows that the duration of a temporal path from $w_{x_\ell}$ to $w_c$ is $2(n-i_\ell+j)$. 
Moreover, the first edge of this path is labeled with 2 and the last edge is labeled with 1.
Both the temporal paths from $v_{x_\ell,1}$ to $u_{c,x_\ell,1}$ and from $v_{x_\ell,2}$ to $u_{c,x_\ell,2}$ have length $2(n-i_\ell+j)+2$. It follows that they can have at most one pair of consecutive edges that are labeled with the same label.
Since by assumption, we have that $\lambda(\{w_{x_\ell},v_{x_\ell,2}\})=2$, we must have that $\lambda(\{w_{c},u_{c,x_\ell,2}\})=2$, since otherwise the temporal path from $v_{x_\ell,2}$ to $u_{c,x_\ell,2}$ has two pairs of consecutive edges (the first two and the last two) that have the same label and this would violate the upper bound $D^{(2)}_{v_{x_\ell,2},u_{c,x_\ell,2}}=D^{(2)}_{u_{c,x_\ell,2},v_{x_\ell,2}}= 2(n-i_\ell+j)+3$. Now the upper bound $D^{(2)}_{u_{c,x_\ell,1},u_{c,x_\ell,2}}=D^{(2)}_{u_{c,x_\ell,1},u_{c,x_\ell,2}}=2$ forces by \Cref{cor:forcing} that $\lambda(\{w_{c},u_{c,x_\ell,1}\})=1$.

Note that if $\lambda(\{w'_x,w_x\})=2$, the main difference is that the first edge of the path from $w_{x_\ell}$ to $w_c$ is labeled with 1 and the last edge is labeled with 2. Since by assumption, we have that $\lambda(\{w_{x_\ell},v_{x_\ell,1}\})=1$, we must have that $\lambda(\{w_{c},u_{c,x_\ell,1}\})=1$, since otherwise the temporal path from $v_{x_\ell,2}$ to $u_{c,x_\ell,2}$ has two pairs of consecutive edges (the first two and the last two) that have the same label and this would violate the upper bound $D^{(2)}_{v_{x_\ell,1},u_{c,x_\ell,1}}=D^{(2)}_{u_{c,x_\ell,1},v_{x_\ell,1}}= 2(n-i_\ell+j)+3$. Now the upper bound $D^{(2)}_{u_{c,x_\ell,1},u_{c,x_\ell,2}}=D^{(2)}_{u_{c,x_\ell,1},u_{c,x_\ell,2}}=2$ forces by \Cref{cor:forcing} that $\lambda(\{w_{c},u_{c,x_\ell,2}\})=2$.

From now on the argument is essentially the same as in the proof of \Cref{lem:corrtd2}.
Since we have $D^{(2)}_{u_{c,x_1,2},u_{c,1}}=D^{(2)}_{u_{c,1},u_{c,x_1,2}}=2$, by \Cref{cor:forcing} we have that $\lambda(\{w_c,u_{c,1}\})=1$.
Now consider the third set of upper bounds for the clause gadget corresponding to $c$. We have that $D^{(2)}_{u_{c,x_2,1},u_{c,x_2,4}}=D^{(1)}_{u_{c,x_2,4},u_{c,x_2,1}}=4$. Since $\lambda(\{w_c,u_{c,x_2,1}\})=\lambda(\{w^\star,u_{c,1}\})=1$, we must have that $\lambda(\{u_{c,1},u_{c,x_2,4}\})=2$ to respect the above upper bound. 
Furthermore, we have that $D^{(2)}_{u_{c,x_2,4},u_{c,2}}=D^{(2)}_{u_{c,2},u_{c,x_2,4}}=2$. By \Cref{cor:forcing} we must have that $\lambda(\{u_{c,1},u_{c,2}\})=1$.

By an analogous argument for the labels on edges in the clause gadget corresponding to $c$ that are associated with $x_3$, we get that $\lambda(\{u_{c,2},u_{c,3}\})=1$.
Summarizing, we get that 
\begin{equation*}
\lambda(\{w_c,u_{c,1}\})=\lambda(\{u_{c,1},u_{c,2}\})=\lambda(\{u_{c,2},u_{c,3}\})=1.
\end{equation*}
This implies that the duration of a fastest path from $w_c$ to $u_{c,3}$ in $(G^{(2)},\lambda)$ is $5$. However, we have $D^{(2)}_{w_c,u_{c,3}}=4$, a contradiction to the assumption that $\lambda$ is a solution for $(G^{(2)},D^{(2)},\Delta)$.
\end{proof}

Now we have all ingredients to prove \Cref{thm:NPh}.

\begin{proof}[Proof of \Cref{thm:NPh}]
The first statement of \Cref{thm:NPh} follows from \Cref{prop:star}.
For the second and third statements, consider the following. Recall that we set $\Delta=2$.
\Cref{lem:corrtd1,lem:corrtd2} together with \Cref{obs:td,obs:red1} prove the second statement of \Cref{thm:NPh}, namely that \deltaUpperBound is NP-hard even if the input tree $G$ has constant diameter and $\Delta=2$.
\Cref{lem:corrmd1,lem:corrmd2} together with \Cref{obs:md,obs:red2} prove the third statement of \Cref{thm:NPh}, namely that \deltaUpperBound is NP-hard even if the input tree $G$ has constant maximum degree and $\Delta=2$.
\end{proof}

As we now show, the hardness reduction on trees for~$\Delta\geq 3$ is tight in the sense that for~$\Delta= 2$, the problem is polynomial-time solvable on trees.
More generally, we show the following.

\begin{theorem}\label{thm:fptvc}
For~$\Delta = 2$, \deltaUpperBound can be solved in $4^{\vc} \cdot n^{O(1)}$~time, where~$\vc$ denotes the vertex cover number of the input tree.
\end{theorem}
\begin{proof}
Let~$I:= (G=(V,E),D,\Delta)$ be an instance of of \deltaUpperBound with~$\Delta = 2$.
Note that we can assume that~$D$ is symmetrical, since for~$\Delta=2$, the duration of a path~$P$ and the duration of its reverse path are identical.
Moreover, we can assume that the for each vertex pair, the duration required by~$D$ is at least the distance between these vertices in~$G$, as otherwise, $I$ is a trivial no-instance.
Additionally, we can assume that~$G$ contains at least one internal vertex, as otherwise, $G$ contains at most two vertices and can thus be solved in polynomial time.

Let~$E'$ denote the internal edges of~$G$, that is, edges $\{u,v\}\in E$ such that neither $u$ nor $v$ is a leaf in $G$. 
Note that~$E'$ has size at most two times the vertex cover number of~$G$.\footnote{
This can be seen as follows:
Let~$Y$ denote a smallest vertex cover of~$G$ and let~$v^*$ be an arbitrary internal vertex of~$G$.
We define~$M$ to be the set of vertices of~$Y$ plus all vertices that are parents of vertices of~$Y$, when rooting the tree~$G$ in vertex~$v^*$.
Clearly, $M$ has size at most~$2\cdot |Y| = 2 \cdot \vc$.
Now, for each internal vertex~$v$, $v$ is a vertex of~$Y$, or some parent of some vertex in~$Y$.
In other words, each internal vertex is contained in~$M$, which implies that~$G[M]$ contains at most~$|M|-1 < 2 \cdot \vc$ edges, since~$G$ is a tree.} 
Our algorithm works as follows:
We iterate over all possible labelings~$\lambda' \colon E' \to \{1,2\}$ and then check whether we can assign labels to the remaining edges (namely, the edges incident with the leaves of~$G$), such that a resulting labeling~$\lambda$ realizes the desired upper bounds of~$D$.
If this is true for some labeling~$\lambda'$, we output the found labeling~$\lambda$.
Otherwise, we output that~$I$ is a no-instance.
Surely, this algorithm is correct.
It thus remains to show the running time of $4^{\vc} \cdot n^{O(1)}$~time.

Since there are only~$2^{|E'|} \leq 2^{2\cdot \vc}$ such labelings~$\lambda'$, it suffices to show that for a given labeling~$\lambda' \colon E' \to \{1,2\}$, we can decide in polynomial time whether there is a labeling~$\lambda \colon E \to \{1,2\}$ that agrees with~$\lambda$ on the labels of all edges of~$E'$.
Let~$\lambda' \colon E' \to \{1,2\}$ and let $G'$ denote the subgraph of $G$ induced on the internal vertices, where a vertex of $G$ is internal if it is not a leaf.
First, we check whether for each two distinct internal vertices~$u$ and~$v$, the duration of the unique temporal~$(u,v)$-path in~$(G',\lambda')$ is at most as specified in~$D$.
If this is not the case, we correctly detect that there is no extension~$\lambda$ of~$\lambda'$ that realizes~$D$.
So suppose that the above property holds for every two distinct internal vertices.
Hence, it remains to ensure that all paths are fast enough between vertex pairs for which at least one vertex is a leaf.
To solve this task, we construct a 2-SAT formula~$F$ and show that~$F$ is satisfiable if and only if there is an extension~$\lambda$ of~$\lambda'$ that realizes~$D$. 

The formula~$F$ contains one variable~$x_e$ for each edge~$e \in E \setminus E'$ (that is, for each edge incident with a leaf).
Intuitively, the truth value of~$x_e$ should be identical to the label that edge~$e$ shall receive (modulo 2).

Since we only checked whether all durations between any two internal vertices are realized, we only have to ensure the durations between vertex pairs for which (i)~one vertex is an internal vertex and one vertex is a leaf or (ii)~both vertices are leaves.

Let~$u$ be an internal vertex, let~$v$ be a leaf, and let~$v'$ be the unique neighbor of~$v$ in~$G$.
If~$u$ and~$v$ are neighbors in~$G$, each labeling realizes their desired duration.
Thus assume that~$u$ and~$v$ are not adjacent.
Moreover, let~$P$ be the subpath from~$u$ to~$v'$ and let~$v''$ be the unique neighbor of~$v'$ in~$P$.
Note that all edges of~$P$ are internal edges.
Let~$q := D(u,v)-d(P)$, where~$d(P)$ is the duration of~$P$ under~$\lambda'$.
If~$q \leq 0$, then no extension of~$\lambda'$ can realize the desired duration on the path from~$u$ to~$v$.
If~$q \geq 2$, then each extension of~$\lambda'$ realizes the desired duration on the path from~$u$ to~$v$.
If~$q = 1$, then the desired duration on the path from~$u$ to~$v$ is realized in exactly those extensions~$\lambda$ of~$\lambda'$ with~$\lambda(\{v',v''\}) \neq \lambda(\{v',v\})$.
To ensure this, we add the clause $(x_{\{v,v'\}})$ to~$F$ if~$\lambda'(\{v',v''\}) = 2$, and we add the clause~$(\neg x_{\{v,v'\}})$ to~$F$, otherwise.
By the above argumentation, this is correct and ensures that in each satisfying assignment for~$F$, the truth value of~$x_{\{v,v'\}}$ is distinct from the label of~$\{v,v'\}$ (modulo 2).

Let~$u$ and~$v$ be leaves and let~$u'$ and~$v'$ be their unique neighbors in~$G$.
We essentially do the same as above.

\textbf{Case 1:}~$u' = v'$\textbf{.}
Since~$u$ and~$v$ have distance two, $D(u,v) \geq 2$.
If~$D(u,v) \geq 3$, then each extension of~$\lambda'$ realizes the desired duration on the path from~$u$ to~$v$.
If~$q = 2$, then the desired duration on the path from~$u$ to~$v$ is realized in exactly those extensions~$\lambda$ of~$\lambda'$ with~$\lambda(\{u,u'\}) \neq \lambda(\{v',v\})$.
To ensure this, we add the clauses $(x_{\{u,u'\}} \lor x_{\{v',v\}})$ and~$(\neg x_{\{u,u'\}} \lor \neg x_{\{v',v\}})$ to~$F$.
By the above argumentation, this is correct and ensures that in each satisfying assignment for~$F$, the truth value of~$x_{\{u,u'\}}$ and~$x_{\{v',v\}}$ are distinct.

\textbf{Case 2:}~$u' \neq v'$\textbf{.}
Let~$P$ be the unique subpath from~$u'$ to~$v'$ in~$G$.
Note that each edge of~$P$ is an internal edge and that~$P$ contains at least one edge, since~$u' \neq v'$. 
Let~$u''$ be the unique neighbor of~$u'$ in~$P$ and let~$v''$ be the unique neighbor of~$v'$ in~$P$.
Let~$q := D(u,v)-d(P)$, where~$d(P)$ is the duration of~$P$ under~$\lambda'$.
If~$q \leq 1$, then no extension of~$\lambda'$ can realize the desired duration on the path from~$u$ to~$v$.
If~$q \geq 4$, then each extension of~$\lambda'$ realizes the desired duration on the path from~$u$ to~$v$.

\textbf{Case 2.1:} $q = 2$\textbf{.}
Then, the desired duration on the path from~$u$ to~$v$ is realized in exactly those extensions~$\lambda$ of~$\lambda'$ with~$\lambda(\{u',u''\}) \neq \lambda(\{u',u\})$ and~$\lambda(\{v',v''\}) \neq \lambda(\{v',v\})$.
To ensure this, we add the clause $(x_{\{u,u'\}})$ to~$F$ if~$\lambda'(\{u',u''\}) = 2$, and we add the clause~$(\neg x_{\{u,u'\}})$ to~$F$, otherwise.
Additionally, we add the clause $(x_{\{v,v'\}})$ to~$F$ if~$\lambda'(\{v',v''\}) = 2$, and we add the clause~$(\neg x_{\{v,v'\}})$ to~$F$, otherwise.
By the above argumentation, this is correct and ensures that in each satisfying assignment for~$F$, (i)~the truth value of~$x_{\{u,u'\}}$ is distinct from the label of~$\{u,u'\}$ (modulo 2) and (ii)~the truth value of~$x_{\{v,v'\}}$ is distinct from the label of~$\{v,v'\}$ (modulo 2).

\textbf{Case 2.2:} $q = 3$\textbf{.}
Then, the desired duration on the path from~$u$ to~$v$ is realized in exactly those extensions~$\lambda$ of~$\lambda'$ with~$\lambda(\{u',u''\}) \neq \lambda(\{u',u\})$ or~$\lambda(\{v',v''\}) \neq \lambda(\{v',v\})$.
To ensure this, we add the clause $$C := \begin{cases}
(x_{\{u,u'\}} \lor x_{\{u,u'\}}) & \lambda'(\{u',u''\}) = 2, \lambda'(\{v',v''\}) = 2 \\
(x_{\{u,u'\}} \lor \neg x_{\{u,u'\}}) & \lambda'(\{u',u''\}) = 2, \lambda'(\{v',v''\}) = 1 \\
(\neg x_{\{u,u'\}} \lor x_{\{u,u'\}}) & \lambda'(\{u',u''\}) = 1, \lambda'(\{v',v''\}) = 2 \\
(\neg x_{\{u,u'\}} \lor \neg x_{\{u,u'\}}) & \lambda'(\{u',u''\}) = 1, \lambda'(\{v',v''\}) = 1 
\end{cases}$$ to~$F$.
By the above argumentation, this is correct and ensures that in each satisfying assignment for~$F$, the truth value of~$x_{\{u,u'\}}$ is distinct from the label of~$\{u,u'\}$ (modulo 2), or the truth value of~$x_{\{v,v'\}}$ is distinct from the label of~$\{v,v'\}$ (modulo 2).

This completes the definition of~$F$.
Note that~$F$ contains at most~$n^2$ clauses and can be constructed in polynomial time and recall that an instance of 2-SAT can be solved in linear time~\cite{aspvall1979linear}.
Hence, by the above argumentation, we can correctly check in polynomial time whether there is an extension~$\lambda$ of~$\lambda'$ that realizes all durations of~$D$.
This completes the proof.
\end{proof}

\section{Main FPT-algorithm for \deltaUpperBound Based on MILPs}\label{sec:FPT}

Our hardness results from the previous section exclude FPT-algorithms for general~$\Delta$ for nearly all reasonable parameters on trees like maximum degree or vertex cover number. 
In this section, we consider the essentially only remaining parameter on trees: the number of leaves of the tree. 
This parameter is larger than the maximum degree and incomparable to the vertex cover number.

\begin{theorem}\label{thm:FPT}
    \deltaUpperBound\ is fixed-parameter tractable when parameterized by the number $\ell$ of leaves of the input tree $G$.
\end{theorem}

To show \Cref{thm:FPT} we present a Turing-reduction from \deltaUpperBound\ to \textsc{Mixed Integer Linear Program} (\textsc{MILP}).

\optproblemdef{\textsc{Mixed Integer Linear Program} (\textsc{MILP})}{A vector $x$ of $n$ variables of which some are considered integer variables, a constraint matrix $A\in\mathbb{R}^{m\times n}$, and two vectors $b\in\mathbb{R}^m$, $c\in \mathbb{R}^n$.}{Compute an assignment to the variables (if one exists) such that all integer variables are set to integer values, $Ax\le b$, $x\ge 0$, and $c^\intercal x$ is maximized.}

Given an instance $(G,D,\Delta)$ of \deltaUpperBound, we will produce several \textsc{MILP} instances. We will prove that $(G,D,\Delta)$ is a yes-instance if and only if at least one of the \textsc{MILP} instances admits a feasible solution. The number of produced instances is upper-bounded by a function of the number $\ell$ of leaves in $G$.
Each of the produced \textsc{MILP} instances will have a small number of integer variables. More precisely, the number of integer variables will be upper-bounded by a function of the number $\ell$ of leaves in $G$. This will allow us to upper-bound the running time necessary to solve the \textsc{MILP} instances using the following known result.

\begin{theorem}[\hspace{-0.0001cm}\cite{Lenstra1983Integer,dadush2011enumerative}]\label{thm:MILP}
    \textsc{MILP} is fixed-parameter tractable when parameterized by the number of integer variables.
\end{theorem}

Furthermore, we build our \textsc{MILP} formulations in a specific way that ensures that there always exist optimal solutions where \emph{all} variables are set to integer values. Informally, we ensure that the constraint matrix for the rational variables is totally unimodular. This allows us to use the following result.

\begin{lemma}[\hspace{-0.0001cm}\cite{CMZ24}]\label{lem:MILP}
    Let $A_{\text{frac}} \in \mathbb{R}^{m \times n_2}$ be totally unimodular.
    Then for any $A_{\text{int}} \in \mathbb{R}^{m \times n_1}$, $b \in \mathbb{R}^m$, and $c \in \mathbb{R}^{n_1 + n_2}$, the MILP
    \(
    \max c^\intercal x \text{ subject to } (A_{\text{int}} \ A_{\text{frac}})x\le b, x\ge 0,
    \)
    where $x = (x_{\text{int}} \ x_{\text{frac}})^\intercal$ with the first $n_1$ variables (i.e., $x_{\text{int}}$) being the integer variables, has an optimal solution where all variables are integer.
\end{lemma}

The main idea for the reduction is the following. Given an instance $(G,D,\Delta)$ of \deltaUpperBound, we create variables for the travel delays at the vertices of $G$. A labelling can easily be computed from the travel delays. We use the constraints to ensure that the durations of the fastest paths respect the upper bounds given in $D$.

Before describing the reduction, we make a straightforward observation. 

\begin{observation}\label{obs:trees}
    Let $G$ be a tree with $\ell$ leaves. Let $V_{>2}$ be the set of vertices in $G$ with degree larger than two. We have that $|V_{>2}|\le \ell$. Furthermore, we have for ever $v\in V_{>2}$ that the degree of $v$ is at most $\ell$.
\end{observation}

Furthermore, we show that there always exist solutions to yes-instances of \deltaUpperBound, where the two edges incident with \degtwo vertices have different labels.
\begin{lemma}\label{lem:labels}
    Let $(G,D,\Delta)$ be a yes-instance of \deltaUpperBound. Let $V_2$ be the set of vertices in~$G$ with degree two. Then there exists a solution $\lambda$ for $(G,D,\Delta)$ such that for all $v\in V_2$ we have the following.
    Let $u,w$ denote the two neighbors of $v$ in $G$, then $\lambda(\{u,v\})\neq \lambda(\{v,w\})$.
\end{lemma}
\begin{proof}
    Assume there is a solution $\lambda$ for $(G,D,\Delta)$ such that for some \degtwo vertex $v$ in $G$ with neighbors $u,w$ we have that $\lambda(\{u,v\})= \lambda(\{v,w\})$. Then we modify $\lambda$ as follows. Let $G_u$ be the subtree of $G$ rooted at $v$ that contains $u$. For each edge $e$ in $G_u$, if $\lambda(e)=\Delta$ the we set $\lambda(e)=1$ and otherwise we increase $\lambda(e)$ by one. By \Cref{def:travelDelays} we have that this leaves all travel delays except $\tau_v^{u,w}$ and $\tau_v^{w,u}$ unchanged and it decreases $\tau_v^{u,w}$ and $\tau_v^{w,u}$ each by at least one. It follows no duration of a fastest temporal path increases. By iterating this procedure, we can obtain a solution $\lambda'$ with the desired property.
\end{proof}

From now on assume that we are given an instance $(G,D,\Delta)$ of \deltaUpperBound. Assume the vertices in $G$ are ordered in an arbitrary but fixed way. Each \textsc{MILP} instance we create will use the same set of variables.
For each vertex $v$ in $G$, we create the following variables:
\begin{itemize}
    \item If $v$ has degree two, we create a \emph{fractional} variable $x_v$. This variable will correspond to the travel delay $\tau_v^{u,w}$, where $u,w$ are the two neighbors of $v$ in $G$ and $u$ is ordered before~$w$.
    \item If $v$ has degree larger than two, then for every pair $u,w$ of neighbors of $v$ in $G$, we create an \emph{integer} variable $y_v^{u,w}$. This variable will correspond to the travel delay $\tau_v^{u,w}$.

    Furthermore, we create an \emph{integer} variable $z_e$ for every edge $e$ that is incident with $v$. This variable will correspond to the label $\lambda(e)$ of $e$.
\end{itemize}
With \Cref{obs:trees} we can upper-bound the number of created integer variables.
\begin{observation}\label{obs:vars}
    Each \textsc{MILP} instance has $O(\ell^3)$ integer variables.
\end{observation}

Now we consider all possibilities of how the labels of edges incident with vertices of degree larger than two can relate to each other. Formally, let $v$ be a vertex in $G$ that has degree larger than two and let $E_v$ be the set of edges incident with $v$. By \Cref{obs:trees} we have that $|E_v|\le\ell$. Any labeling $\lambda$ partitions $E_v$ into at most $\min(\ell,\Delta)$ sets of edges with equal labels, and the values of the labels define an ordering of those sets. We call a partitioning of~$E_v$ into at most $\min(\ell,\Delta)$ sets together with the ordering for those sets a \emph{label configuration} for~$v$. A set of label configurations for all vertices $v$ in $G$ that have degree larger than two is called a \emph{global label configuration}. By \Cref{obs:trees} we get the following.
\begin{observation}\label{obs:configs}
    There are $O(\ell^{\ell^2})$ global label configurations.
\end{observation}

For each global label configuration, we create an \textsc{MILP} instance. From now on fix a global label configuration $\sigma$. We describe how to construct an \textsc{MILP} instance $I_\sigma$. Consider a vertex $v$ in $G$ that has degree larger than two. Let $u,w$ be a pair of neighbors of $v$ in $G$. Let $e=\{u,v\}$ and $e'=\{v,w\}$. 
If according to the global label configuration, we have that $e$ and $e'$ are in different parts of $E_v$ and the part of $e$ is ordered before the part of $e'$ (that is,~$e$ has a smaller label than~$e'$), then we add the following constraint.
\begin{equation}\label{constr:1}
    y_v^{u,w} = z_{e'}-z_e
\end{equation}
 If the part of $e$ is ordered after the part of $e'$ (that is, $e$ has a larger label than~$e'$), then we add the following constraint.
\begin{equation}\label{constr:2}
    y_v^{u,w} = z_{e'}-z_e+\Delta
\end{equation}
In both above cases, we additionally add the following constraint.
\begin{equation}\label{constr:3}
    1\le y_v^{u,w}\le \Delta-1
\end{equation}

If according to the global label configuration, we have that $e$ and $e'$ are in the same part of~$E_v$ (that is, they have the same label), then we add the following constraints.
\vspace{-0.5cm}

\begin{center}
$%
\begin{array}{cccc}
 \ \ z_e = z_{e'}\text{ \ (\ref{constr:4})}\refstepcounter{equation}\label{constr:4}& 
\hspace{0.9cm} y_v^{u,w} = \Delta\text{ \ (\ref{constr:5})}\refstepcounter{equation}\label{constr:5} & 
\hspace{0.9cm} 1\le z_e \le \Delta\text{ \ (\ref{constr:6})}\refstepcounter{equation}\label{constr:6} & 
\hspace{0.9cm} 1\le z_{e'} \le \Delta\text{ \ (\ref{constr:7})}\refstepcounter{equation}\label{constr:7}%
\end{array}%
$
\end{center}

Now we consider all vertex pairs $s,t$ in $G$ that have distance at least two in~$G$, that is, the (unique) path from $s$ to $t$ in $G$ has at least one internal vertex. 
Let $P_{s,t}$ denote the path from $s$ to $t$ in $G$.
Then $V^{(s,t)}_{\text{int}}=V(P_{s,t})\setminus\{s,t\}$ is the set of internal vertices of $P_{s,t}$. 
We partition $V^{(s,t)}_{\text{int}}$ into two sets. Let $V^{(s,t)}_{\text{int}, \text{deg}2}\subseteq V^{(s,t)}_{\text{int}}$ be the internal vertices of $P_{s,t}$ with degree two and let $V^{(s,t)}_{\text{int}, \text{deg}>2}\subseteq V^{(s,t)}_{\text{int}}$ be the internal vertices with degree larger than two. Now we further partition the two sets of vertices. 

We define $V^{(s,t)}_{\text{int}, \text{deg}2, \text{f}}\subseteq V^{(s,t)}_{\text{int}, \text{deg}2}$ such that for all $v\in V^{(s,t)}_{\text{int}, \text{deg}2, \text{f}}$, the neighbor of $v$ that is closer to $s$ is ordered before the neighbor of $v$ that is closer to $t$.
Analogously, we define $V^{(s,t)}_{\text{int}, \text{deg}2, \text{b}}\subseteq V^{(s,t)}_{\text{int}, \text{deg}2}$ such that for all $v\in V^{(s,t)}_{\text{int}, \text{deg}2, \text{b}}$, the neighbor of $v$ that is closer to $t$ is ordered before the neighbor of $v$ that is closer to $s$.

For the vertices in $V^{(s,t)}_{\text{int}, \text{deg}>2}$ we additionally define the following. For $v\in V^{(s,t)}_{\text{int}, \text{deg}>2}$ let~$v'$ denote the neighbor of $v$ that is on the path from $v$ to $s$ and let $v''$ denote the neighbor of $v$ that is on the path from $v$ to $t$.
We define $u^{(s,t)}(v)=v'$ and $w^{(s,t)}(v)=v''$.

We add the following constraint for the vertex pair $s,t$.
\begin{equation}\label{constr:8}
    \sum_{v\in V^{(s,t)}_{\text{int}, \text{deg}2, \text{f}}} x_v \ + \sum_{v\in V^{(s,t)}_{\text{int}, \text{deg}2, \text{b}}} (\Delta - x_v) \ +  \sum_{v\in V^{(s,t)}_{\text{int}, \text{deg}>2}} y^{u^{(s,t)}(v),w^{(s,t)}(v)}_v \ \le \ D_{s,t}-1
\end{equation}

Finally, for all fractional variables $x_v$ we add the following constraint.
\begin{equation}\label{constr:9}
    1\le x_v\le \Delta-1
\end{equation}

This finishes the description of the constraints. Since we only want to check whether a feasible solution exists, we arbitrarily choose the vector $c$ for the objective function to be the all-one vector. This finishes the description of the \textsc{MILP} instance $I_\sigma$. Clearly, for a given global label configuration, the \textsc{MILP} instance can be constructed in polynomial time. 
\begin{observation}\label{obs:instance}
    Given a global label configuration $\sigma$, the \textsc{MILP} instance $I_\sigma$ can be computed in polynomial time.
\end{observation}

In the following, we prove the correctness of the algorithm. 
We first show that if the given instance of \deltaUpperBound\ is a yes-instance, then at least one of the created \textsc{MILP} instances has a feasible solution.
\begin{lemma}
\label{lem:corr1}
    Given a yes-instance $(G,D,\Delta)$ of \deltaUpperBound, there exists a global label configuration~$\sigma$ such that the \textsc{MILP} instance $I_\sigma$ admits a feasible solution.
\end{lemma}
\begin{proof}
    Let $(G,D,\Delta)$ be a yes-instance of \deltaUpperBound. Let $V_2$ be the set of vertices in $G$ with degree two. Then by \Cref{lem:labels} there exists a solution $\lambda$ for $(G,D,\Delta)$ such that for all $v\in V_2$ we have the following. Let $u,w$ denote the two neighbors of $v$ in $G$. We have that $\lambda(\{u,v\})\neq \lambda(\{v,w\})$. 

    Furthermore, let $\sigma$ be the global label configuration defined by the labeling function $\lambda$. We construct a feasible solution for the \textsc{MILP} instance $I_\sigma$ as follows.

    First, consider all edges $e$ that are incident with vertices of degree larger than two. We set~$z_e=\lambda(e)$. Note that, since Constraints~(\ref{constr:1}),~(\ref{constr:2}), and~(\ref{constr:5}) are equality constraints, there is a uniquely defined value for each variable $y_v^{u,w}$ that we need to assign to that variable in order to fulfill those constraints. Furthermore, since $\lambda$ defines the global label configuration~$\sigma$, we also have that the remaining Constraints~(\ref{constr:3}),~(\ref{constr:4}),~(\ref{constr:6}), and~(\ref{constr:7}) are fulfilled by the assignment of the variables. Lastly, it is obvious that all integer variables are assigned integer values.

    Now we assign values to the fractional variables. Let $V_2$ denote the set of all \degtwo vertices in $G$. We have a variable $x_v$ for each $v\in V_2$. Consider some $v\in V_2$ and let $u,w$ denote the two neighbors of $v$, where $u$ is ordered before~$w$. Then we set 
    \begin{equation*}
        x_v = \begin{cases}
        \lambda (\{v,w\}) - \lambda(\{u,v\}), & \text{if } \lambda (\{v,w\}) > \lambda(\{u,v\}),\\
        \lambda (\{v,w\}) - \lambda(\{u,v\}) + \Delta, & \text{otherwise}.
        \end{cases}
    \end{equation*}
    Recall that due to \Cref{lem:labels}, we have that $\lambda(\{u,v\})\neq \lambda(\{v,w\})$. It follows that Constraint~(\ref{constr:9}) is fulfilled. It remains to show that Constraint~(\ref{constr:8}) is fulfilled.

    To this end, consider a vertex pair $s,t$ and let $P=\left(\left(v_{i-1},v_i,t_i\right)\right)_{i=1}^k$ be a fastest temporal $(s,t)$-path in the $\Delta$-periodic temporal graph $(G,\lambda)$. Since $G$ is a tree, we have that $V(P)$ is uniquely determined. In fact, also the order in which the vertices in $V(P)$ are visited by $P$ is uniquely determined.
    By \Cref{lem:duration} we have that
        \begin{equation*}
        d(P)=1+\sum_{i\in[k-1]}\tau_{v_i}^{v_{i-1},v_{i+1}}.
    \end{equation*}
    Furthermore, we have the following.
    \begin{itemize}
        \item If $v_i$ has degree two in $G$ and $v_{i-1}$ is ordered before $v_{i+1}$, then we have $x_{v_i}=\tau_{v_i}^{v_{i-1},v_{i+1}}$.
        \item If $v_i$ has degree two in $G$ and $v_{i+1}$ is ordered before $v_{i-1}$, then by \Cref{obs:traveldelays} we have $x_{v_i}=\Delta-\tau_{v_i}^{v_{i-1},v_{i+1}}$.
        \item If $v_i$ has degree larger than two in $G$, then we have $y^{v_{i-1},v_{i+1}}_{v_i}=\tau_{v_i}^{v_{i-1},v_{i+1}}$.
    \end{itemize}
    Lastly, since $\lambda$ is a solution for $(G,D,\Delta)$, we have that $d(P)\le D_{s,t}$. From this, it follows that Constraint~(\ref{constr:8}) is fulfilled. Hence, we have constructed a feasible solution to the \textsc{MILP} instance $I_\sigma$.
\end{proof}

Next, we show that if one of the produced \textsc{MILP} instances admits a feasible solution, then we are facing a yes-instance of \deltaUpperBound. To this end, we first show that the constraint matrix for the fractional variables of the produced \textsc{MILP} instances is totally unimodular. Then by \Cref{lem:MILP} we have that if one of the \textsc{MILP} instances admits a feasible solution, then there also is a feasible solution for that instance where all fractional variables are set to integer values. 
To do this, we use a sufficient condition for matrices to be totally unimodular. Precisely, we use that every so-called \emph{network matrix} is totally unimodular.
\begin{definition}[network matrix]\label{def:networkmatrix}
     Let $T=(V,A)$ be a directed tree, that is, a tree where each arc has an arbitrary orientation and let $A'$ be a set of directed arcs on $V$, the same vertex set as $T$. Let $M$ be a matrix with $|A|$ rows and $|A'|$ columns. Then $M$ is a \emph{network matrix} if the following holds.
     For all $i,j$ let $e_i=(a,b)\in A$ and $e'_j=(s,t)\in A'$.
     \begin{itemize}
         \item $M_{i,j}=1$ if arc $(a,b)$ appears forwards in the path in $T$ from $s$ to $t$.
         \item $M_{i,j}=-1$ if arc $(a,b)$ appears backwards in the path in $T$ from $s$ to $t$.
         \item $M_{i,j}=0$ if arc $(a,b)$ does not appear in the path in $T$ from $s$ to $t$.
     \end{itemize}
\end{definition}
\begin{lemma}[\hspace{-0.0001cm}\cite{schrijver2003combinatorial}]\label{lem:TU}
    Network matrices are totally unimodular.
\end{lemma}

We show that the produced \textsc{MILP} instances have the following property.
\begin{lemma}\label{lem:fracTU}
    Let $\sigma$ be a global label configuration and let $I_\sigma$ be the corresponding \textsc{MILP} instance.
    Let $M$ be the constraint matrix for the fractional variables, that is, the constraint matrix obtained by treating the integer variables $y_v^{u,w}$ and~$z_e$ in all constraints of $I_\sigma$ as arbitrary constants, deleting all constraints that do not involve any fractional variables, and deleting all duplicate rows.
    We have that $M$ is a network matrix.
\end{lemma}
\begin{proof}
    First, notice that Constraints~(\ref{constr:1}),~(\ref{constr:2}),~(\ref{constr:3}),~(\ref{constr:4}),~(\ref{constr:5}),~(\ref{constr:6}), and~(\ref{constr:7}) do not involve fractional variables. Hence, from now on we only consider Constraints~(\ref{constr:8}) and~(\ref{constr:9}), where in Constraint~(\ref{constr:8}) we treat the integer variables as constants. Let $M$ be the resulting constraint matrix where also all duplicate rows are deleted. We show that $M$ is a network matrix.

\begin{figure}[t]
\noindent\makebox[\textwidth]{
\centering
\ifarxiv
\includegraphics[scale=1]{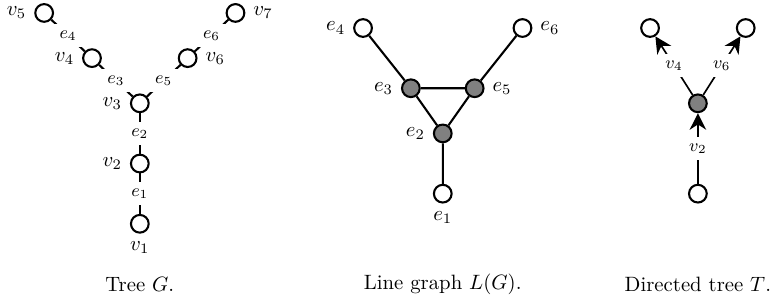}
\else
\scalebox{.9}{
\begin{tikzpicture}[line width=1pt,scale=.6,xscale=.9,yscale=.85] 
    
    \node[vert,label=below:$v_1$] (V1) at (0,0) {}; 
    \node[vert,label=left:$v_2$] (V2) at (0,2) {}; 
    \node[vert,label=left:$v_3$] (V3) at (0,4) {}; 
    \node[vert,label=left:$v_4$] (V4) at (-1.5,5.5) {}; 
    \node[vert,label=left:$v_5$] (V5) at (-3,7) {}; 
    \node[vert,label=right:$v_6$] (V6) at (1.5,5.5) {}; 
    \node[vert,label=right:$v_7$] (V7) at (3,7) {}; 

    \node (G) at (0,-2) {Tree $G$.};
    
    \draw (V1) --node[timelabel] {$e_1$} (V2);
    \draw (V2) --node[timelabel] {$e_2$} (V3);
    \draw (V3) --node[timelabel] {$e_3$} (V4);
    \draw (V4) --node[timelabel] {$e_4$} (V5);
    \draw (V3) --node[timelabel] {$e_5$} (V6);
    \draw (V6) --node[timelabel] {$e_6$} (V7);

    \node[vert,label=below:$e_1$] (U1) at (9.5,1) {};
    \node[vert,label=left:$e_2$,fill=gray] (U2) at (9.5,3) {};
    \node[vert,label=left:$e_3$,fill=gray] (U3) at (8.5,4.5) {};
    \node[vert,label=left:$e_4$] (U4) at (7,6.5) {};
    \node[vert,label=right:$e_5$,fill=gray] (U5) at (10.5,4.5) {};
    \node[vert,label=right:$e_6$] (U6) at (12,6.5) {};

    \node (LG) at (9.5,-2) {Line graph $L(G)$.};

    \draw (U1) -- (U2);
    \draw (U2) -- (U3);
    \draw (U3) -- (U4);
    \draw (U2) -- (U5);
    \draw (U5) -- (U6);
    \draw (U3) -- (U5);

    \node (T) at (17.5,-2) {Directed tree $T$.};

    \node[vert] (W1) at (17.5,1) {};
    \node[vert,fill=gray] (W2) at (17.5,4) {};
    \node[vert] (W3) at (16,6.5) {};
    \node[vert] (W4) at (19,6.5) {};

    \draw[diredge] (W1) --node[timelabel] {$v_2$} (W2);
    \draw[diredge] (W2) --node[timelabel] {$v_4$} (W3);
    \draw[diredge] (W2) --node[timelabel] {$v_6$} (W4);
\end{tikzpicture}
}
\fi}
    \caption{Illustration of the construction of the directed tree $T$ in \Cref{lem:fracTU}. On the left an example tree $G$, where the vertices are ordered by their indices. In the middle the line graph $L(G)$ of $G$, where the gray vertices form a maximal clique with more than two vertices. On the right the constructed directed tree~$T$, where the vertices of the maximal clique are merged into one vertex (gray) and the directed arcs correspond to the \degtwo vertices of $G$.}\label{fig:network}
\end{figure}

    To this end, we define a directed tree $T=(V,A)$ (see \Cref{def:networkmatrix}) as follows. We give an illustration in \Cref{fig:network}. 
    Let $L(G)$ be the line graph of the input tree $G$ of the \deltaUpperBound\ instance. 
    Now we exhaustively merge all maximal cliques in $L(G)$ with more than two vertices to a single vertex.
    Here, the new vertex is adjacent to all remaining neighbors of the originally merged vertices.
    Let this graph be called $G'$. It is easy to see that $G'$ is a tree, and every edge in~$G'$ corresponds to a \degtwo vertex in $G$, and vice versa. Furthermore, we have that every vertex~$v$ in~$G'$ that has degree larger than two corresponds to a subtree $G_v$ (which may be the degenerate tree consisting only of one vertex) in $G$ such that all vertices in $V(G_v)$ have degree larger than two in $G$.
    Now we transform $G'$ into $T$ by replacing each edge of $G'$ with an oriented arc. Consider a \degtwo vertex $v$ of $G$ and the corresponding edge $e_v$ in $G'$. Let $u$ and $w$ be the two neighbors of $v$ in $G$, where $u$ is ordered before~$w$. We distinguish four cases. 
    \begin{itemize}
        \item Vertex $u$ has degree two in $G$. Let $e_u$ be the edge in $G'$ corresponding to $u$. Then we orient $e_v$ such that the resulting arc points away from the common endpoint of $e_v$ and $e_u$.
        \item The degree of vertex $u$ does not equal two and $w$ has degree two in $G$. Let $e_w$ be the edge in $G'$ corresponding to $w$. Then we orient $e_v$ such that the resulting arc points towards the common endpoint of $e_v$ and $e_w$.
        \item Neither $u$ nor $w$ have degree two. Then we can assume that at least one of $u$ and $w$ has degree larger than two, otherwise $G$ only has three vertices and we can solve the instance in constant time. 
        \begin{itemize}
            \item Vertex $u$ has degree larger than two in $G$. Then let $G_u$ denote the maximal subtree of~$G$ that contains $u$ such that all vertices in $V(G_u)$ have degree larger than two in $G$. Then there is a vertex $u'$ in $G'$ corresponding to $G_u$ that has degree larger than two and is incident with $e_v$. We orient $e_v$ such that the resulting arc points away from $u'$.
            \item Vertex $u$ has degree one and vertex $w$ has degree larger than two in $G$. Then let $G_w$ denote the maximal subtree of $G$ that contains $w$ such that all vertices in $V(G_w)$ have degree larger than two in $G$. Then there is a vertex $w'$ in $G'$ corresponding to $G_w$ that has degree larger than two and is incident with $e_v$. We orient $e_v$ such that the resulting arc points towards $w'$.
        \end{itemize}
    \end{itemize}
    This finishes the description of the directed tree $T$.

    Next, we make the following observation: Each Constraint~(\ref{constr:9}) produces the same coefficients for the rational variables in the constraint matrix as one of Constraint~(\ref{constr:8}). To see this, note that each Constraint~(\ref{constr:9}) has exactly one non-zero coefficient for one rational variable~$x_v$, which is either $1$ or $-1$. Variable $x_v$ corresponds to a \degtwo vertex $v$ in $G$. Let $u$ and~$w$ be the two neighbors of $v$ in $G$, where $u$ is ordered before~$w$. Then, since $u$ and $w$ have distance at least two in $G$, we have a Constraint~(\ref{constr:8}) corresponding to the pair of vertices~$s,t$ with $s=u$ and $t=w$. This constraint has exactly one non-zero coefficient for a rational variable, which is $x_v$ and the coefficient is $1$. In the case where $s=w$ and $t=u$ we have the same situation, except that the coefficient is $-1$. Hence, since $M$ has no duplicate rows, we can assume that each row of $M$ corresponds to a Constraint~(\ref{constr:8}).

    Now consider a row of $M$ corresponding to Constraint~(\ref{constr:8}), which in turn corresponds to vertex pair $s,t$ in $G$, where $s$ and $t$ have distance at least two in $G$. Let $P_{s,t}$ denote the (unique) path from $s$ to $t$ in $G$. Let $V^{(s,t)}_{\text{int}, \text{deg}2, \text{f}}\subseteq V(P_{s,t})$ and $V^{(s,t)}_{\text{int}, \text{deg}2, \text{b}}\subseteq V(P_{s,t})$ be defined as in the description of Constraint~(\ref{constr:8}). Then we have that $V^{(s,t)}_{\text{int}, \text{deg}2, \text{f}}\cup V^{(s,t)}_{\text{int}, \text{deg}2, \text{b}}$ are all internal vertices of $P_{s,t}$ that have degree two in $G$. Furthermore, in Constraint~(\ref{constr:8}) corresponding to vertex pair $s,t$ we have that the fractional variable $x_v$ has a non-zero coefficient if and only if $v\in V^{(s,t)}_{\text{int}, \text{deg}2, \text{f}}\cup V^{(s,t)}_{\text{int}, \text{deg}2, \text{b}}$. More specifically, the coefficient of $x_v$ is $1$ if $v\in V^{(s,t)}_{\text{int}, \text{deg}2, \text{f}}$ and the coefficient of $x_v$ is $-1$ if $v\in V^{(s,t)}_{\text{int}, \text{deg}2, \text{b}}$.

    Furthermore, by construction we have that each vertex in $V^{(s,t)}_{\text{int}, \text{deg}2, \text{f}}\cup V^{(s,t)}_{\text{int}, \text{deg}2, \text{b}}$ corresponds to an arc of $T$. It remains to show that there is a path $P'$ in $T$ that visits exactly the arcs corresponding to the vertices in $V^{(s,t)}_{\text{int}, \text{deg}2, \text{f}}\cup V^{(s,t)}_{\text{int}, \text{deg}2, \text{b}}$ such that the arcs corresponding to vertices in $V^{(s,t)}_{\text{int}, \text{deg}2, \text{f}}$ appear forwards in the path and the arcs corresponding to the vertices in $V^{(s,t)}_{\text{int}, \text{deg}2, \text{b}}$ appear backwards in the path.
    To this end, we order the vertices in $V^{(s,t)}_{\text{int}, \text{deg}2, \text{f}}\cup V^{(s,t)}_{\text{int}, \text{deg}2, \text{b}}$ in the order in which they are visited by $P_{s,t}$, that is, let $V^{(s,t)}_{\text{int}, \text{deg}2, \text{f}}\cup V^{(s,t)}_{\text{int}, \text{deg}2, \text{b}} = \{v_1, v_2, \ldots, v_{k}\}$. Now we show by induction on $k$ that $P'$ exists and that its last arc corresponds to $v_k$. If $k=1$ then $P'_{s,t}$ only consists of the arc $e_{v_1}$ corresponding to $v_1$. If $v_1\in V^{(s,t)}_{\text{int}, \text{deg}2, \text{f}}$, then we consider $P'$ to be the path in which $e_{v_1}$ forwards. Otherwise, we consider $P'$ to be the path in which $e_{v_1}$ backwards. Now assume $k>1$. Let $P''$ be the path in $T$ that visits the arcs corresponding to vertices $\{v_1,\ldots,v_{k-1}\}$ such that for all $i\in[k-1]$ the arc corresponding to $v_i$ appears forwards if $v_i\in V^{(s,t)}_{\text{int}, \text{deg}2, \text{f}}$ and backwards otherwise. Furthermore, we have that $P''$ ends with the arc corresponding to $v_{k-1}$. Now we show that we can append the arc corresponding to $v_k$ to $P''$.
    We distinguish two cases. 
    \begin{itemize}
        \item In the first case, $v_k$ is visited directly after $v_{k-1}$ by $P_{s,t}$. Then the arcs $e_{v_{k-1}}$ and $e_{v_k}$ in~$T$ share a common endpoint which is a \degtwo vertex. It follows that we can append the arc $e_{v_k}$ to $P''$ to create a new path that ends in $e_{v_k}$. We still need to show that $e_{v_k}$ is appended in the correct direction. Consider the two neighbors of $v_k$ in $G$. One of the neighbors is $v_{k-1}$. Let the second neighbor be $w$. Now assume that $v_k\in V^{(s,t)}_{\text{int}, \text{deg}2, \text{f}}$. This implies that $v_{k-1}$ is ordered before $w$. In this case the arc $e_{v_k}$ in $T$ points away from the common endpoint of $e_{v_{k-1}}$ and $e_{v_k}$, and hence it appears forwards in the $P''$ appended with $e_{v_k}$. The case where $v_k\in V^{(s,t)}_{\text{int}, \text{deg}2, \text{b}}$ is analogous.
    
        \item In the second case, there is a non-empty set of vertices $\widehat{V}$ with degree larger than two, that are visited by $P_{s,t}$ between $v_{k-1}$ and $v_k$. Let $\widehat{G}$ denote the maximal tree in $G$ that only contains vertices that have degree larger than two in $G$ and that contains $\widehat{V}$. By construction, there is vertex $\widehat{v}$ in $T$ that corresponds to $\widehat{G}$ and both $e_{v_{k-1}}$ and $e_{v_k}$ in $T$ have $\widehat{v}$ as one of their endpoints. First, we argue that we can append $e_{v_k}$ to $P''$ and thereby create a path in $T$ that ends with $e_{v_k}$. To this end, notice that the vertex visited by $P_{s,t}$ before $v_{k-1}$ cannot be part of $\widehat{G}$, otherwise $G$ would contain a cycle. It follows that the common endpoint of $e_{v_{k-2}}$ and $e_{v_{k-1}}$ is not $\widehat{v}$ and appending $e_{v_k}$ to $P''$ indeed produces a path in $T$ that ends with $e_{v_k}$. It remains to show that $e_{v_k}$ is appended in the correct direction. Consider the two neighbors $u$ and $w$ of $v_k$ in $G$. We have that one neighbor is contained in $\widehat{G}$, assume w.l.o.g.\ that it is $u$. Then $u$ is visited by $P_{s,t}$ before~$v_{k}$. Now assume that $v_k\in V^{(s,t)}_{\text{int}, \text{deg}2, \text{f}}$. This implies that $u$ is ordered before $w$. In this case the arc $e_{v_k}$ in $T$ points away from $u$, and hence it appears forwards in the path $P''$ appended with $e_{v_k}$. The case where $v_k\in V^{(s,t)}_{\text{int}, \text{deg}2, \text{b}}$ is analogous.
    \end{itemize}
It follows that for every row in $M$, there is a path in $T$ such that the conditions in \Cref{def:networkmatrix} are met. We can conclude that $M$ is a network matrix.
\end{proof}

Now we are ready to show that if one of the produced \textsc{MILP} instances admits a feasible solution, then we are facing a yes-instance of \deltaUpperBound. 

\begin{lemma}
\label{lem:corr2}
    Given an instance $(G,D,\Delta)$ of \deltaUpperBound, if there exists a global label configuration~$\sigma$ such that the \textsc{MILP} instance $I_\sigma$ admits a feasible solution, then $(G,D,\Delta)$ is a yes-instance of \deltaUpperBound.
\end{lemma}
\begin{proof}
Let $(G,D,\Delta)$ be an instance of \deltaUpperBound. Assume there exists a global label configuration~$\sigma$ such that the \textsc{MILP} instance $I_\sigma$ admits a feasible solution. By \Cref{lem:MILP,lem:TU,lem:fracTU} and the fact that totally unimodularity is preserved under row duplication~\cite{schrijver2003combinatorial} we have that then there exists a feasible solution where \emph{all} variables are set to integer values. Assume we are given such a solution for $I_\sigma$. We construct a solution $\lambda$ for $(G,D,\Delta)$ as follows.

To this end, we make the following observation. Let $u,v,w$ be three vertices in $G$ such that $\{u,v\}$ and $\{v,w\}$ are edges in $G$ and $u$ is ordered before $w$. Then for each fixed label $\lambda(\{u,v\})$ there is a uniquely defined label $\lambda(\{v,w\})$ such that $\tau_v^{u,w}=x_v$ if $v$ has degree two in $G$, and $\tau_v^{u,w}=y_v^{u,w}$ if $v$ has degree larger than two in $G$.

This allows us to produce a labeling $\lambda$ with the following procedure. We arbitrarily pick an edge $e$, say the lexicographically smallest one according to the vertex ordering, in $G$ and set $\lambda(e)=1$. 
Now we iteratively pick an edge $\{u,v\}$ that has not received a label yet and shares a common endpoint $v$ with an edge $\{v,w\}$ in $G$ that has already received a label. By the observation above, the label of $\{v,w\}$ together with the travel delay $\tau_v^{u,w}$ uniquely defines the label of $\{u,v\}$, and in turn $\tau_v^{u,w}$ is given by $x_v$, $y_v^{u,w}$, or $y_v^{w,u}$, depending on the degree of $v$ and whether $u$ is ordered before $w$ or not. We claim that the described procedure produces a solution $\lambda$ for $(G,D,\Delta)$.

First, we investigate the travel delays in the $\Delta$-periodic temporal graph $(G,\lambda)$. By construction and \Cref{obs:traveldelays}, we have for each vertex $v$ that has degree two in $G$, and neighbors $u,w$, where~$u$ is ordered before $w$, that $\tau_v^{u,w}= x_v$ and $\tau_v^{w,u}=\Delta - x_v$. For each vertex $v$ in $G$ that has degree larger than two, the above described procedure does not per se guarantee that for each pair $u,w$ of neighbors of $v$ that $\tau_v^{u,w}= y_v^{u,w}$. However, notice that the labeling $\lambda$ is uniquely determined by the label used for the first edge. It follows that there are $\Delta$ different labelings that can be produced by the above described procedure, one for each label that can be put on the first edge. Call them $\lambda=\lambda_1,\lambda_2,\ldots,\lambda_\Delta$, where for all edges $e$ and for all $1\le i<j\le \Delta$ we have $\lambda_i(e)=\lambda_j(e)+i-j$. In particular, this means that for each vertex $v$ that has degree larger than two and each edge $e$ that is incident with $v$, there is an $i\in[\Delta]$ such that $\lambda_i(e)=z_e$. In this case, Constraints~(\ref{constr:1}),~(\ref{constr:2}),~(\ref{constr:3}),~(\ref{constr:4}),~(\ref{constr:5}),~(\ref{constr:6}), and~(\ref{constr:7}) ensure that indeed for all edges $e$ that are incident with $v$ we have that $\lambda_i(e)=z_e$. Furthermore, these constraints ensure that indeed for each pair $u,w$ of neighbors of $v$ we have $\tau_v^{u,w}= y_v^{u,w}$. Since all labelings $\lambda_1,\ldots,\lambda_\Delta$ produce the same travel delays, we can conclude that for each vertex $v$ in $G$ that has degree larger than two and for each pair $u,w$ of neighbors of $v$ in $(G,\lambda)$ we have $\tau_v^{u,w}= y_v^{u,w}$.

Assume for contradiction that $\lambda$ is not a solution for $(G,D,\Delta)$. Then there are two vertices $s,t$ in $G$ such that the duration of a fastest temporal path $P_{s,t}=\left(\left(v_{i-1},v_i,t_i\right)\right)_{i=1}^k$ in the $\Delta$-periodic temporal graph $(G,\lambda)$ is larger than $D_{s,t}$. By \Cref{lem:duration} we have that 
    \begin{equation*}
        d(P_{s,t})=1+\sum_{i\in[k-1]}\tau_{v_i}^{v_{i-1},v_{i+1}}.
    \end{equation*}
As discussed above, we have the following:
\begin{itemize}
    \item If $v_i$ has degree two and $v_{i-1}$ is ordered before $v_{i+1}$, then $\tau_{v_i}^{v_{i-1},v_{i+1}}=x_{v_i}$.
    \item If $v_i$ has degree two and $v_{i+1}$ is ordered before $v_{i-1}$, then $\tau_{v_i}^{v_{i-1},v_{i+1}}=\Delta - x_{v_i}$.
    \item If $v_i$ has degree larger than two, then $\tau_{v_i}^{v_{i-1},v_{i+1}}=y_{v_i}^{v_{i-1},v_{i+1}}$.
\end{itemize}
Hence, using the notation introduced for Constraint~(\ref{constr:8}), we can write
\begin{equation*}
    d(P_{s,t})=1+\sum_{v\in V^{(s,t)}_{\text{int}, \text{deg}2, \text{f}}} x_v \ + \sum_{v\in V^{(s,t)}_{\text{int}, \text{deg}2, \text{b}}} (\Delta - x_v) \ +  \sum_{v\in V^{(s,t)}_{\text{int}, \text{deg}>2}} y^{u^{(s,t)}(v),w^{(s,t)}(v)}_v.
\end{equation*}
However, since Constraint~(\ref{constr:8}) is fulfilled for the vertex pair $s,t$, it follows that $d(P_{s,t})\le D_{s,t}$, a contradiction. This finishes the proof.
\end{proof}

\Cref{thm:FPT} follows from the shown results. 
\begin{proof}[Proof of \Cref{thm:FPT}]
    We use the following algorithm to solve \deltaUpperBound. Let $(G,D,\Delta)$ be an instance of \deltaUpperBound. For each global label configuration $\sigma$ we create the \textsc{MILP} instance $I_\sigma$. If for some $\sigma$ the instance $I_\sigma$ admits a feasible solution, we answer ``yes''. Otherwise, we answer~``no''.

    By \Cref{lem:corr1,lem:corr2} this algorithm is correct. It remains to bound the running time. Let $\ell$ be the number of leaves of $G$. By \Cref{obs:configs} we create $O(\ell^{\ell^2})$ \textsc{MILP} instance. By \cref{obs:instance} we can create each \textsc{MILP} instance in polynomial time and by \Cref{obs:vars} each created \textsc{MILP} instance has $O(\ell^3)$ integer variables. By \Cref{thm:MILP} it follows that we can solve each \textsc{MILP} instance in FPT-time with respect to parameter $\ell$. We can conclude that the overall running time of the algorithm is in $f(\ell)\cdot |(G,D,\Delta)|^{O(1)}$ for some computable function $f$. 
\end{proof}

\section{Conclusion}
We have initiated the investigation of the natural temporal tree realization problem \deltaUpperBound and shown that it NP-hard in quite restrictive cases. 
On the positive side, we provided an FPT-algorithm for the number of leaves in the input tree as a parameter, essentially the only reasonable parameter on trees that is not a constant in at least one of our hardness reductions. 
A canonical future research direction is to investigate \deltaUpperBound on general graphs.
For example, can our FPT-algorithm for number of leaves can be transferred to general graphs? 
Further interesting parameters for general graphs include the distance to a clique (or distance to other dense graphs), maximum independent set, clique cover, or even various parameters of the input matrix $D$.

\bibliography{bib}

\end{document}